\documentclass[letterpaper]{article}
\usepackage{jair}
\usepackage{lscape}
\usepackage{graphicx}
\usepackage{makeidx}
\usepackage{amsmath, amsxtra, amssymb, amsfonts, mathabx}
\usepackage{xcolor}
\usepackage{tikz}
\usetikzlibrary{positioning}
\usepackage{comment}

\newcommand{\quitar}[1]{}
\newcommand{\myneg}{\overline}
\newcommand{\F}{{\cal F}}
\newcommand{\G}{{\cal G}}
\newcommand{\R}{{\cal H}}
\newcommand{\K}{{\cal K}}
\newcommand{\B}{{\cal B}}
\newcommand{\A}{{\cal A}}
\newcommand{\C}{{\cal C}}

\newtheorem{property}{Property}
\newtheorem{example}{Example}
\newtheorem{lemma}{Lemma}
\newtheorem{theorem}{Theorem}
\newtheorem{definition}{Definition}
\newtheorem{corollary}{Corollary}

\newenvironment{proof}{\noindent{\bf Proof.}}{$\quad$\newline}

\title{Towards a Better Understanding of (Partial Weighted) MaxSAT Proof Systems
\thanks{Research funded by FEDER/Ministerio de Ciencia e Innovaci{\'{o}}n $-$ Agencia Estatal de Investigaci{\'{o}}n, Spain, with project RTI2018-094403-B-C33,
}
}
\author{Javier Larrosa \and Emma Rollon}

\author{\name{Javier Larrosa} \email larrosa@cs.upc.edu \\
        \addr Universitat Polit$\grave{e}$cnica de Catalunya \\
              Barcelona, Spain
        \AND
        \name Emma Rollon \email erollon@cs.upc.edu \\
        \addr Universitat Polit$\grave{e}$cnica de Catalunya \\
              Barcelona, Spain
}

\date{}

\begin{document}

\maketitle

\begin{abstract}
    MaxSAT, the optimization version of the well-known SAT problem, has attracted a lot of research interest in the last decade. Motivated by the many important applications and
    inspired by the success of modern SAT solvers,  researchers have developed many MaxSAT solvers. Since most research is algorithmic, its significance is mostly evaluated empirically.
    In this paper we want to address MaxSAT from the more formal point of view of Proof Complexity. With that aim we start providing basic definitions and proving some basic results. Then
    we analyze the effect of adding split and virtual, two original inference rules, to MaxSAT resolution. We show that each addition makes the resulting proof system stronger, with the virtual rule capturing the recently proposed concept of circular proof.
\end{abstract}

\section{Introduction}
\textit{Proof complexity} is the field aiming to understand the computational cost required to prove or refute statements. Different proof systems may provide different proofs for the same formula and some proof systems are provably more efficient than others.  When that happens, proof complexity cares about which elements of the more powerful proof system really make the difference.
 
 In propositional logic,  resolution-based proof systems that work with CNF formulas have attracted the interest of researchers for several decades \cite{DBLP:journals/jsyml/Buresh-OppenheimP07}. One of the reasons is that CNF is the working language of the extremely successful SAT solvers and the most essential ingredients of these algorithms (e.g, conflict analysis) can be understood and analyzed as proofs \cite{handbookSAT}.
 
 (Partial Weighted) \textit{MaxSAT} is the optimization version of SAT. Since many discrete optimization problems are naturally represented as MaxSAT problems, 
 %Although discrete optimization problems can be modeled and solved with SAT solvers, many of these problems are more naturally treated as MaxSAT. For this reason 
 the design of MaxSAT solvers has attracted the interest of researchers. Interestingly, while some of the first efficient MaxSAT solvers were strongly influenced by MaxSAT inference \cite{DBLP:journals/ai/LarrosaHG08}, this influence has diminished along time. The currently most efficient algorithms solve MaxSAT by sophisticated sequences of calls to SAT solvers\cite{DBLP:journals/constraints/MorgadoHLPM13,DBLP:journals/ai/AnsoteguiBL13,DBLP:conf/ijcai/BacchusHJS18}. %We find intriguing that state-of-the-art MaxSAT solvers, instead of depending on MaxSAT resolution, make internal calls to SAT solvers (which depend on SAT resolution), while it has been recently observed that a clever use of MaxSAT resolution is more efficient than SAT resolution for the refutation of pure SAT formulas~\cite{DBLP:conf/aaai/BonetBIMM18}. 
 
 The purpose of this paper is to improve our understanding of resolution-based MaxSAT proof systems. This is important at least for two reasons. One is to understand if there is some fundamental explanation for SAT-based MaxSAT solvers being superior to MaxSAT-resolution-based MaxSAT solvers. Another reason is to better understand the advantages and disadvantages of different inference rules which, in turn, can shed some light on the power of MaxSAT resolution and help produce better solvers.
 
 This paper contributes in both directions\footnote{A preliminary version of this paper appear in \cite{DBLP:conf/aaai/LarrosaR20} and \cite{SATLarrosa20}.}. First, we extend some classic proof complexity concepts (i.e, entailment, completeness, etc) to facilitate a proof complexity approach to MaxSAT. One interesting result is that, similarly to what happens in SAT, refutational completeness makes completeness somehow redundant or, in other words, that a MaxSAT solver can be used to prove or disprove entailment. We also introduce \textit{split} and \textit{virtual}, two new MaxSAT inference rules that complement MaxSAT resolution. We show that each add-on makes a 
 %provable 
 stronger system. More precisely, we show that: the proof system containing only resolution (\textbf{Res}) is sound and refutationally complete; adding the split rule (\textbf{ResS}) we get completeness and (unlike what happens in SAT) exponential speed-up in certain refutations; further adding the virtual rule (\textbf{ResSV}), which allows to keep negative weights during proofs, we get further exponential speed-up by capturing the concept of \textit{circular proofs} \cite{DBLP:conf/sat/AtseriasL19}. It is known that SAT circular proofs can efficiently refute the Pigeon Hole Principle. We show that \textbf{ResSV} can refute \textit{hard} and \textit{soft} versions of the Pigeon Hole Principle. From our work we also get the interesting and unexpected result that in some cases rephrasing a SAT refutation as a MaxSAT entailment may transform the problem from exponentially hard to polynomial when using \textbf{ResSV}.
 
 Figure \ref{fig:resumen} summarizes the main contributions of the paper in terms of comparing the different proof systems. The top row considers the general case of MaxSAT proofs and the bottom row considers the particular case of SAT refutations (i.e, refutation of MaxSAT formulas with hard clauses, only) as considered in \cite{DBLP:conf/sat/IgnatievMM17,DBLP:journals/eccc/FilmusMSV20,DBLP:conf/aaai/BonetBIMM18,DBLP:conf/sat/BonetL20}.

The structure of the paper is as follows: in Section~\ref{Sec-Background} we provide preliminaries on SAT and MaxSAT. In Section~\ref{Sec-PHP} we define some variations of the Pigeon Hole Problem that we need for the proofs of the theorems. In Section~\ref{Sec-SATPS} we review some concepts on SAT proof systems in order to facilitate their extension to MaxSAT, which is provided in Section~\ref{Sec-MaxPS}. 
In Section~\ref{Sec-MaxRules} we present, discuss and analyze the three proof systems: \textbf{Res}, \textbf{ResS} and \textbf{ResSV}.
In Section~\ref{Sec-Circular} we show how the strongest proof system \textbf{ResSV} captures the notion of Circular Proof. In Section \ref{Sec-Related} we contextualize our work with some previous related works and finally, in Section~\ref{Sec-Conclusions}, we give some conclusions.

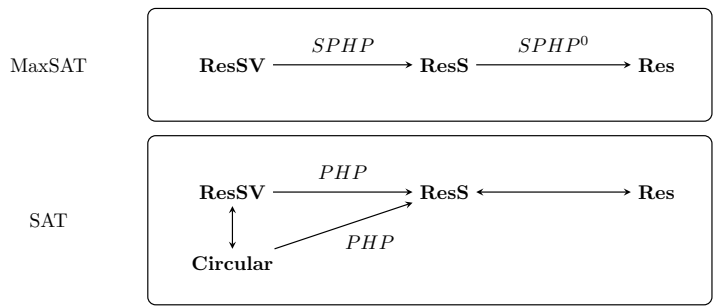
\begin{figure}[t]
\begin{center}
\scalebox{0.75}{\begin{tikzpicture}

\tikzset{%
every node/.style={fill=white, rounded corners, node distance=5em},%
every path/.style={->, >=stealth, line width=.06em}%
}

\tikzstyle{op}=[draw, fill=black!5!white]

\node [draw,rectangle, minimum width=10cm, minimum height=2cm, fill=none] (maxsat) {};
\node[left of = maxsat, xshift=-5cm] {MaxSAT};

\node [below of=maxsat, yshift=-1cm,draw,rectangle, minimum width=10cm, minimum height=3cm, fill=none] (sat) {};
\node[left of = sat, xshift=-5cm] {SAT};

\node[overlay, yshift=0cm, xshift=-3.5cm](ressv){\textbf{ResSV}};
\node[overlay,right of=ressv, xshift=2cm](ress){\textbf{ResS}};
\node[overlay,right of=ress, xshift=2cm](res){\textbf{Res}};

\draw[overlay] (ressv) -- node[overlay,above, midway, yshift=0.1cm, align=center] {$SPHP$} (ress) ;
\draw[overlay] (ress) -- node[overlay,above, midway, yshift=0.1cm, align=center] {$SPHP^0$} (res) ;

\node[overlay, below of=ressv, yshift=-0.5cm](satressv){\textbf{ResSV}};
\node[overlay, below of=ress, yshift=-0.5cm](satress){\textbf{ResS}};
\node[below of=res, yshift=-0.5cm](satres){\textbf{Res}};
\node[overlay, below of=satressv, yshift=0.5cm](circular){\textbf{Circular}};
\draw[overlay,<->] (satressv) to (circular) ;

\draw[overlay] (satressv) -- node[overlay, above, midway, yshift=0.1cm, align=center] {$PHP$} (satress) ;
\draw[overlay, <->] (satress) to (satres) ;
\draw[overlay] (circular) -- node[overlay, right, yshift=-0.3cm, xshift=-0.1cm,align=center] {$PHP$} (satress) ;

\end{tikzpicture}}
    \caption{Comparison among different proof systems. An $\textbf{S}\longrightarrow \textbf{S'}$ means that $\textbf{S}$ p-simulates $\textbf{S'}$, but $\textbf{S'}$ does not p-simulate $\textbf{S}$. The problem that proves that the p-simulation is not in both directions is indicated over the arrow. In the lower rectangle the comparison is restricted to SAT refutations (i.e, input formulas with hard clauses only).
    }
    \label{fig:resumen}
\end{center}
\end{figure}

\section{Background}\label{Sec-Background}

\subsection{SAT Problem}
A \textit{boolean variable} $x$ takes values on the set $\{0,1\}$. A \textit{literal} is a variable $x$ (positive literal) or its negation $\myneg x$ (negative literal). 
We will use sets of literals to denote variable assignments (a.k.a. truth assignments) with literals with $x$ (respectively $\myneg x)$ representing that variable $x$ is instantiated with $1$ (respectively $0$).  
A \textit{clause} is a disjunction of literals. A clause $C$ is satisfied by a truth assignment $X$ if $X$ contains at least one of the literals in $C$. The \textit{empty clause} is denoted $\Box$ and cannot be satisfied. 

A \textit{CNF formula} $\F$ is a set of clauses (taken as a conjunction). A truth assignment satisfies a formula if it satisfies all its clauses. If such an assignment exists, we say that the assignment is a \textit{model} and the formula is \textit{satisfiable}. We say that formula $\F$ \textit{entails} formula $\G$, noted $\F\models \G$, if every model of $\F$ is also a model of $\G$. Two formulas $\F$ and $\G$ are \textit{equivalent}, noted $\F\equiv \G$, if they entail each other. 

Given a formula $\F$, the SAT problem, noted $SAT(\F)$, is to determine if $\F$ is satisfiable or not.
The negation of a clause $C=l_1\lor l_2 \lor \ldots \lor l_p$ is satisfied if all its literals are falsified and this can be trivially expressed in CNF as the set of unit clause $\myneg C= \{\myneg l_1, \myneg l_2, \ldots, \myneg l_p\}$.

\subsection{MaxSAT Problem}
A \textit{weight} $w$ is a non-negative integer or $\infty$ (i.e, $w\in \mathbb{N}\cup \{\infty\}$). We extend addition and subtraction to weights defining $\infty + w =\infty$ and $\infty - w=\infty$ for all $w$. Note that $v-w$ is only defined when $w\leq v$. 

A \textit{weighted clause} is a pair $(C,w)$ where $C$ is a clause and $w$ is a weight associated to its falsification. If $w=\infty$ we say that the clause is \textit{hard}, else it is \textit{soft}.

A \textit{weighted MaxSAT CNF formula} is a multiset of weighted clauses $\F= \{(C_1, w_1), \ldots, (C_p, w_p)\}$. If all the clauses are hard, we say that the formula is \textit{hard}. If all the clauses are soft, we say that the formula is \textit{soft}. Otherwise the formula is \textit{mixed}. Unless we explicitly say otherwise, we will assume mixed formulas. This definition of MaxSAT including soft and hard clauses is sometimes referred to as Partial Weighted MaxSAT \cite{DBLP:journals/constraints/MorgadoHLPM13} and corresponds to the most general MaxSAT language.

%Any clause $(C,w)$ is equivalent to two clauses $(C,u), (C,v)$ as long as $u+v=w$. In the following we will assume that clauses are unmerged and merged as needed. We say that $\G \subseteq \F$ if for all $(C,w)\in \G$ there is a $(C,w')\in \F$ with $w\leq w'$. 

Given a formula $\F$, we define the cost of a truth assignment $X$, noted $\F(X)$, as the sum of weights over the clauses that are falsified by $X$. We say that formula $\F$ \textit{entails} formula $\G$, noted $\F\models \G$, if for all $X$, $\G(X)$ is a lower bound of $\F(X)$ (i.e., $\forall X$, $\F(X)\geq \G(X)$). We say that two formulas $\F$ and $\G$ are \textit{equivalent}, noted $\F\equiv \G$, if they entail each other (i.e., $\forall X$, $\F(X)=\G(X)$).  

%Note that if $\G \subseteq \F$ then $\F\models \G$ and there is a formula $\R$ such that $\F\equiv \G \cup \R$.

Given a formula $\F$, the MaxSAT problem, noted $MaxSAT(\F)$,  is to find the minimum cost over the set of all truth assignments,
$$MaxSAT(\F)=\min_X \F(X)$$
Note that if the hard clauses of the formula make it unsatisfiable then $MaxSAT(\F)=\infty$.

In the following sections we will find useful to deal with negated weighted clauses. Hence, the corresponding definitions and useful property. Let $A$ and $B$ be arbitrary disjunctions of literals. Let $(A \lor \myneg B, w)$ mean that falsifying $A \lor \myneg B$ incurs a cost of $w$. Although $A \lor \myneg B$ is not a clause, the following property shows that it can be efficiently transformed into a weighted CNF equivalent.

\begin{property}\label{propNegCl}
$\{(A \lor \myneg{l_1\lor l_2 \lor \ldots \lor l_p} , w)\}\equiv \{(A \lor \myneg l_1, w), (A \lor l_1 \lor \myneg l_2, w), \ldots, (A \lor l_1 \lor \ldots \lor l_{p-1} \lor \myneg l_p, w) \}$.
\end{property}{}

The negation of a MaxSAT formula $\F$ is the negation of all its clauses,
$$\myneg\F=\{(\myneg C, w) \mid (C,w)\in \F\}$$
For example, the negation of
formula $\F=\{(x\lor y,\infty), (\myneg x \lor \myneg y,3)\}$
 is $\myneg \F=\{(\myneg x,\infty), (x\lor \myneg y,\infty), (x,3), (\myneg x \lor y,3)\}$.

\section{Pigeon Hole Problem and Variations}\label{Sec-PHP} 
We define the well-known Pigeon Hole Problem $\mathit{PHP}$ and three MaxSAT soft versions $\mathit{SPHP}$, $\mathit{SPHP}^0$ and $\mathit{SPHP}^1$, that we will be using in the proof of  our results.

In the
 \textit{Pigeon Hole Problem} $\mathit{PHP}$ the goal is to assign $m+1$ pigeons to $m$ holes without any pair of pigeons sharing their hole. In the usual SAT encoding there is a boolean variable $x_{ij}$ (with $1\leq i \leq m+1,$ and $1\leq j \leq m$) which is true if pigeon $i$ is in hole $j$. There are two groups of clauses. For each pigeon $i$, we have the clause,
$${\cal P}_i= \{x_{i1}\lor x_{i2}\lor \ldots \lor x_{im}\}$$
indicating that pigeon $i$ must be assigned to at least one hole. For each hole $j$ we have the set of clauses,
$${\cal H}_j=\{\myneg x_{ij} \lor \myneg x_{i'j} \mid 1\leq i<i'\leq m+1\}$$
indicating that hole $j$ is occupied by at most one pigeon. Let $\K$ be the union  of all these sets of clauses
$\K=\cup_{1\leq i \leq m+1} {\cal P}_i ~\cup_{1\leq j \leq m} {\cal H}_j$. 
It is obvious that $\K$ is an unsatisfiable CNF formula.
In MaxSAT notation the pigeon hole problem is,
$$\mathit{PHP}=\{(C,\infty) \mid C\in \K\}$$
\noindent and clearly $MaxSAT(\mathit{PHP})=\infty$.

In the \textit{soft Pigeon Hole Problem} $\mathit{SPHP}$  the goal is to find the assignment that falsifies the minimum number of clauses. In MaxSAT language it is encoded as,
$$\mathit{SPHP}=\{(C,1) \mid C\in \K\}$$
\noindent and it is obvious that $MaxSAT(\mathit{SPHP}) = 1$.

The $\mathit{SPHP^0}$ problem is like the soft pigeon hole problem  but augmented with one more clause $(\Box, m^2+m)$ where $m$ is the number of holes. 
Note that $MaxSAT(\mathit{SPHP^0}) = m^2+m+1$.

Finally, the $\mathit{SPHP^1}$ problem is like the soft pigeon hole problem  but augmented with a set of unit clauses $\{(x_{ij},1),(\myneg x_{ij},1) \mid 1\leq i\leq m+1, 1\leq j\leq m\}$. 
Note that $MaxSAT(\mathit{SPHP^1}) = m^2+m+1$.

\section{SAT Proof Systems}\label{Sec-SATPS}

A \textit{SAT proof system} $\textbf{S}$ is a set of inference rules. An \textit{inference rule} is given by a set of antecedent clauses and a set of consequent clauses. In SAT, an inference rule means that if the antecedents are members of the formula, the consequents can be \textit{added}. The rule is \textit{sound} if every truth assignment that satisfies the antecedents also satisfies the consequents. 

 A \textit{proof}, or \textit{derivation}, under a proof system $\textbf{S}$ is a finite sequence $C_1, C_2, \ldots , C_e$ where the start of the sequence, $C_1,\ldots , C_p$, is the original formula ${\cal F}$ and each $C_i$ (with $i>p$) is obtained by applying an inference rule from $\textbf{S}$ with earlier antecedents (i.e., $C_j$ with $j<i$). The \textit{length} of the proof is $e-p$. A \textit{polynomial size} proof is a proof whose length can be bounded by a polynomial on $|{\cal F}|$. 
 
 We will write $\F\vdash_{S} \G$ to denote an arbitrary proof $\Pi= (C_1, C_2, \ldots , C_e)$ with $\F=\{C_1,\ldots, C_p\}$ and $\G \subseteq \cup_{i=1}^e \{C_i\}$ (abusing notation, in the following we will note $\G \subseteq \cup_{i=1}^e \{C_i\}$ as $\G \subseteq \Pi$ ). When the proof system is irrelevant or implicit from the context we will just write $\F\vdash \G$. A \textit{refutation} of $\F$ is a proof $\F \vdash_{S} \Box$. Refutations are important because they prove unsatisfiability.
 
 A proof $\Pi=(C_1, C_2, \ldots , C_e)$ can be graphically represented as an acyclic directed bi-partite graph $G(\Pi) = (J \cup I, E)$ such that in $J=\{C_1,...,C_e\}$ and each node in $I$ represents an inference step. Consider the inference step with antecedents ${\cal A}\subset \Pi$ and consequents $\C \subset \Pi$. Node $R\in I$ has ${\cal A}$ in-neighbours and $\C$ out-neighbours. 
 Since the same clause can be derived several times, different nodes in $J$ may correspond to the same clause. 
 Note that clauses in the original formula $\F$  do not have in-neighbors. The rest of the clauses have exactly one in-neighbour. All clauses may have several out-neighbors since they may be used as an antecedent several times during the proof.
 
 \begin{figure}
     \centering
\scalebox{0.75}{\begin{tikzpicture}
\tikzset{%
every node/.style={fill=white, rounded corners, node distance=5em},%
every path/.style={->, >=stealth, line width=.06em}%
}

\tikzstyle{op}=[draw, fill=black!5!white]

\node(xy){$x \lor y$};
\node[right of = xy](-x){$\myneg x$};
\node[right of = -x](-y){$\myneg y$};

\node[op, below of= xy, yshift=0.7cm, xshift=0.8cm](res1){res};
\draw (-x) to (res1);
\draw (xy) to (res1);
\node[below of= res1, yshift=0.7cm](y){y};
\draw (res1) to (y);
\node[op, below of= y, yshift=1cm, xshift=1.5cm](res2){res};
\draw (-y) to (res2);
\draw (y) to (res2);
\node[below of= res2, yshift=0.7cm](box){$\Box$};
\draw (res2) to (box);

\node[below of = box, yshift=-0.5cm](dummy){}; % para centrar
\end{tikzpicture}
\hspace{3cm}
\begin{tikzpicture}
\tikzset{%
every node/.style={fill=white, rounded corners, node distance=5em},%
every path/.style={->, >=stealth, line width=.06em}%
}

\tikzstyle{op}=[draw, fill=black!5!white]

\node(xy){$x \lor y$};
\node[right of = xy, xshift = 1cm](-x){$\myneg x$};
\node[right of = -x, xshift = 1cm](-y){$\myneg y$};

\node[op, below of= -x, yshift=0.7cm](split1){split};
\draw (-x) to (split1);
\node[below=of split1, yshift=1.3cm, xshift=-0.7cm](-xy){$\myneg x \lor y$};
\node[below=of split1, yshift=1.3cm, xshift=0.7cm](-x-y){$\myneg x \lor \myneg y$};
\draw (split1) to (-xy);
\draw (split1) to (-x-y);
\node[op, below of= -xy, yshift=0.7cm, xshift=-0.7cm](res1){sym res};
\draw (-xy) to (res1);
\draw (xy) to (res1);
\node[below of= res1, yshift=0.7cm](y){y};
\draw (res1) to (y);
\node[op, below of= y, yshift=1cm, xshift=2cm](res2){sym res};
\draw (-y) to (res2);
\draw (y) to (res2);
\node[below of= res2, yshift=0.7cm](box){$\Box$};
\draw (res2) to (box);
\end{tikzpicture}}
     \caption{Refutation graph for $\{x\lor y, \myneg x, \myneg y\}$ using the resolution rule (left) and symmetric resolution with split (right). }
     \label{fig:satgraph}
 \end{figure}
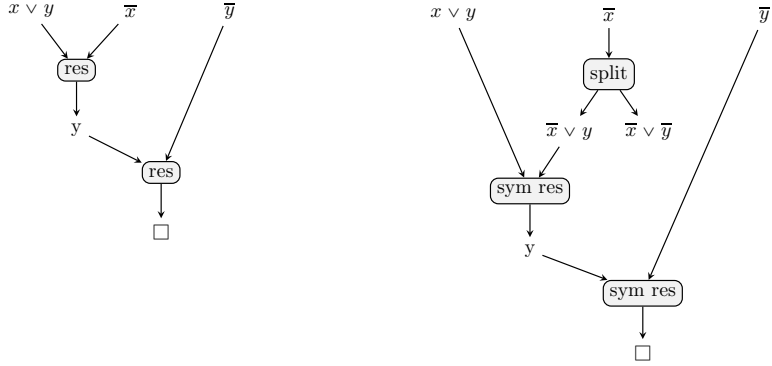
 
A proof system $\textbf{S}$ is \textit{sound} if $\F \vdash_{S} \G$ implies $\F \models \G$. It is \textit{complete} if $\F\models \G$ implies 
$\F\vdash_{S}\G$. 
Although completeness is a natural and elegant property, it has limited practical interest. For that reason a weaker version of completeness has been defined. A proof system $\textbf{S}$ is \textit{refutationally complete} if $\F \models \Box$ implies  $\F\vdash_{S} \Box$. In words, for every unsatisfiable formula $\F$ there is a refutation $\F\vdash_{S} \Box$ (i.e, completeness is required only for refutations). It is usually believed that refutational completeness is enough for practical purposes. The reason is that $\F \models \G$ if and only if $\F\cup \myneg \G \models \Box$ (i.e., $\F\cup \myneg \G$ is unsatisfiable), so any refutationally complete proof system can prove the entailment by deriving $\Box$ from a CNF formula equivalent to $\F\cup \myneg \G$.

The most usual way to compare the strength of different proof systems is with the concept of $p$-simulation. We say that proof system $\textbf{S}$ $p$-simulates proof system $\textbf{S'}$ if there is a polynomially computable function $f$ such that for every $\textbf{S}$-refutation $\Pi$ of formula $\F$, $f(\Pi)$ is an $\textbf{S'}$-refutation of the same formula $\F$. If $\textbf{S}$ $p$-simulates $\textbf{S'}$ and $\textbf{S'}$ does not $p$-simulate $\textbf{S}$ we say that $\textbf{S}$ is stronger or more powerful than $\textbf{S'}$.

Consider the following three sound inference rules \cite{Robinson:1965:MLB:321250.321253} \cite{DBLP:conf/sat/AtseriasL19},

\[
\begin{array}{ccccc}
x\lor A \hspace{1cm} \myneg x \lor B & & x\lor A \hspace{1cm} \myneg x \lor A & &A\\
\cline{1-1} 
\cline{3-3} 
\cline{5-5} 
A\lor B  & & A & & A \lor x \hspace{.5cm} A \lor \myneg x  \\
\texttt{(resolution)} & & \texttt{(symmetric resolution)} & &\texttt{(split)} \\
\end{array}
\]

\noindent where $A$ and $B$ are arbitrary (possibly empty) disjunctions of literals and $x$ is an arbitrary variable. In propositional logic it is customary to define rules with just one consequent because one rule with $s$ consequents can be obtained from $s$ one-consequent rules. As we will see, this is not the case in MaxSAT. For this reason, here we prefer to introduce the two-consequents split rule instead of the equivalent weakening rule \cite{DBLP:conf/sat/AtseriasL19} to keep the parallelism with MaxSAT more evident.

It is well-known that the proof system made exclusively of resolution is refutationally complete and adding the split rule makes the system complete.  However, the following property says that adding the split rule does not give any advantage to resolution in terms of refutational power,

\begin{property}\label{sat-res-ress}[(see Lemma 7 in \cite{DBLP:journals/jacm/Atserias04}]
A proof system with resolution and split as inference rules cannot make shorter refutations than a proof system with only resolution.
\end{property}

It is easy to see that resolution can be simulated by split and symmetric resolution, so the resulting proof system is also complete.
Figure \ref{fig:satgraph} shows a refutation graph of $\{x\lor y, \myneg x, \myneg y\}$ 
using the resolution rule (left) and symmetric resolution with split (right).

\section{MaxSAT Proof Systems and Completeness}
\label{Sec-MaxPS}

A \textit{MaxSAT proof system} $\textbf{S}$ is a set of MaxSAT inference rules. A \textit{MaxSAT inference rule} is given by a set of antecedent clauses and a set of consequent clauses. In MaxSAT, the application of an inference rule is to \textit{replace} the antecedents by the consequents. The process of applying an inference rule to a formula $\F$ is noted $\F;\F'$. The rule is \textit{sound} if it preserves the equivalence of the formula i.e, $\F \equiv \F'$.

A \textit{proof}, or \textit{derivation}, with proof system $\textbf{S}$ is a sequence ${\cal F}_0; {\cal F}_1; \ldots; {\cal F}_e$ where ${\cal F}_0$ is the original formula $\F$ and each ${\cal F}_i$ is obtained by applying an inference rule from $\textbf{S}$. The length of the proof is $e$. Note that MaxSAT proofs are sequences of formulas while SAT proofs are sequences of clauses. We use the semi-colon to emphasize this distinction. The reason is that MaxSAT inference rules modify clauses already in the formula in order to derive new ones, so each step of the proof must carry along the whole formula. Note that a SAT proof $C_1,C_2,\ldots, C_e$ with $\F_0=\{C_1,...,C_p\}$ with comma notation can easily be transformed to the semi-colon notation as $\F_0;\F_1;\ldots;\F_{e-p}$ where each formula $\F_i$ contains the new clauses and all the previous clauses, $\F_i=\{C_1,...,C_{p+i}\}$.

We will write $\F\vdash_{S} \G$ to denote an arbitrary proof ${\cal F}_0; {\cal F}_1; \ldots; {\cal F}_e$ with $\F=\F_0$ and $\G\subseteq \F_e$.
A proof system $\textbf{S}$ is \textit{sound} if $\F \vdash_{S} \G$ implies $\F \models \G$. It is \textit{complete} if $\F\models \G$ implies  $\F\vdash_{S} \G$.
A $k$-\textit{refutation} of $\F$ is a proof $\F \vdash_{S} (\Box,k)$. 
A proof system is \textit{refutationally complete} if there is a proof $\F \vdash_{S} (\Box,k)$ for every formula $\F$ and every $k\leq MaxSAT(\F)$.
Until Section~\ref{Sec-Circular} we will only consider $k$-refutations with $k=MaxSAT(\F)$ and we will refer to them simply as refutations. Section~\ref{Sec-Circular} will consider the special case of $1$-refutations of hard formulas.

\begin{figure}[t]
     \centering
\scalebox{0.75}{\begin{tikzpicture}
\tikzset{%
every node/.style={fill=white, rounded corners, node distance=5em},%
every path/.style={->, >=stealth, line width=.06em}%
}

\tikzstyle{op}=[draw, fill=black!5!white]

\node(xy){$(x \lor y, \infty)$};
\node[right of = xy](-x){$(\myneg x, \infty)$};
\node[right of = -x](-y){$(\myneg y, 1)$};

\node[op, below of= xy, yshift=0.7cm, xshift=0.8cm](res1){res};
\draw (-x) to (res1);
\draw (xy) to (res1);
\node[below of= res1, yshift=0.7cm](-x2){$(\myneg x, \infty)$};
\node[right of= -x2, xshift=-0.6cm](y){$(y, \infty)$};
\node[left of= -x2, xshift=0.4cm](xy2){$(x \lor y, \infty)$};
\draw (res1) to (y);
\draw (res1) to (-x2);
\draw (res1) to (xy2);

\node[op, below of= y, yshift=0.7cm](unmerge){unmerge};
\draw (y) to (unmerge);
\node[below of= unmerge, yshift=0.7cm, xshift=-0.6cm](y2){$(y, \infty)$};
\node[below of= unmerge, yshift=0.7cm, xshift=0.6cm](y3){$(y, 1)$};
\draw (unmerge) to (y2);
\draw (unmerge) to (y3);

\node[op, below of= y3, yshift=1cm, xshift=1cm](res2){res};
\draw (-y) to (res2);
\draw (y3) to (res2);

\node[below of= res2, yshift=0.7cm](box){$(\Box, 1)$};
\draw (res2) to (box);
%\draw (res2) to (y2);

\node[below of=box, yshift=-0.4cm](dummy){}; % para centrar
\end{tikzpicture}
\hspace{3cm}
\begin{tikzpicture}
\tikzset{%
every node/.style={fill=white, rounded corners, node distance=5em},%
every path/.style={->, >=stealth, line width=.06em}%
}

\tikzstyle{op}=[draw, fill=black!5!white]

\node(xy){$(x \lor y, \infty)$};
\node[right of = xy, xshift=0cm](-x){$(\myneg x, \infty)$};
\node[right of = -x](xvirtual){$(x, 1)$};
\node[right of = xvirtual](-y){$(\myneg y, 1)$};
\node[left of = xy, xshift=-0.9cm](x-1){$(x, -1)$};

\node[op, above of = xy, yshift=-0.5cm, xshift=0.8cm](virtual){virtual};
\draw (virtual) to (x-1);
\draw (virtual) to (xvirtual);

\node[op, below of = -x, yshift=0.7cm](unmergex){unmerge};
\draw (-x) to (unmergex);
\node[below of = unmergex, yshift=0.7cm, xshift=-0.6cm](-xinf){$(\myneg x, \infty)$};
\node[below of = unmergex, yshift=0.7cm, xshift=0.6cm](-x1){$(\myneg x, 1)$};
\draw (unmergex) to (-xinf);
\draw (unmergex) to (-x1);

\node[op, below of = -x1, yshift=0.7cm, xshift=0.8cm](resx){res};
\draw (-x1) to (resx);
\draw (xvirtual) to (resx);

\node[below of= resx, yshift=0.8cm](boxx){$(\Box, 1)$};
%\node[below of= resx, yshift=0.8cm, xshift=0.6cm](x2){$(\myneg x, \infty)$};
\draw (resx) to (boxx);
%\draw (resx) to (x2);

\node[op, below of= boxx, yshift=0.7cm](splitbox){split};
\draw (boxx) to (splitbox);
\node[below of= splitbox, yshift=0.8cm, xshift=-0.6cm](-y1){$(\myneg y, 1)$};
\node[below of= splitbox, yshift=0.8cm, xshift=0.6cm](y1){$(y, 1)$};
\draw (splitbox) to (-y1);
\draw (splitbox) to (y1);

\node[op, right of= y1, xshift=-0.2cm](resbox){res};
\draw (y1) to (resbox);
\draw (-y) to (resbox);
\node[below of= resbox, yshift=0.7cm](r){$(\Box, 1)$};
\draw (resbox) to (r);

\node[op, below of= xy, yshift=-1.4cm](unmergexy){unmerge};
\draw (xy) to (unmergexy);
\node[below of= unmergexy, yshift=0.8cm, xshift=-0.9cm](xyinf){$(x \lor y, \infty)$};
\node[below of= unmergexy, yshift=0.8cm, xshift=0.9cm](xy1){$(x \lor y, 1)$};
\draw (unmergexy) to (xyinf);
\draw (unmergexy) to (xy1);

\node[op, below of= -y1, yshift=0.9cm, xshift=-1cm](resy){res};
\draw (-y1) to (resy);
\draw (xy1) to (resy);

\node[below of= resy, yshift=0.7cm](xy2){$(x \lor y, \infty)$};
\node[left of= xy2, xshift=0.3cm](x1){$(x, 1)$};
\node[right of= xy2, xshift=-0.1cm](-x-y){$(\myneg x \lor \myneg y, 1)$};
\draw (resy) to (xy2);
\draw (resy) to (x1);
\draw (resy) to (-x-y);

\node[op, left of= x1, xshift=-1cm](merge){merge};
\draw (x-1) to (merge);
\draw (x1) to (merge);

%\node[below of= merge, yshift=0.7cm](x0){}; %{$(x, 0)$};
%\draw (merge) to (x0);
\end{tikzpicture}}
    \caption{A MaxSAT refutation graph for MaxSAT formula $\{(x\lor y, \infty), (\myneg x, \infty), (\myneg y, 1)\}$ using \textbf{Res} (left) and \textbf{ResSV} (right), respectively.}
    \label{fig:graphMaxSAT}
\end{figure}

As in the SAT case, a MaxSAT proof $\Pi=(\F_0; \F_1; \ldots; F_e)$ can be graphically represented as an acyclic directed bi-partite graph $G(\Pi) = (J \cup I, E)$ where there is one node in $J$ for each clause and one node in $I$ for each inference step. Clauses in $\F_0$ do not have in-neighbours. Consider proof step $\F_{i-1};\F_i$ where the antecedents of the inference are ${\cal A}\subseteq \F_{i-1}$ and the consequents are $\C\subseteq \F_i$. The inference node has ${\cal A}$ as in-neighbors and $\C$ as out-neighbors.

There are two differences with respect to the SAT case: nodes in $J$ contain a clause and a weight, and they have at most one out-neighbor. 
%; and for convenience there is a special type of inference node that represents the merging of equal clauses adding up their weight. 
Figure \ref{fig:graphMaxSAT} shows two refutation graphs for $\{(x\lor y, \infty), (\myneg z, \infty), (\myneg y, 1)\}$ using two different proof systems, to be defined later.

Now we show that, similarly to what happens in SAT, refutationally completeness is sufficient for practical purposes. The reason is that it can also be used to prove  or disprove general entailment, making completeness somehow redundant.
Let the \textit{roof} of a formula $\F$, noted $\mathit{rf}(\F)$, be the sum of its weights,
$$\mathit{rf}(\F)=\sum_{(C,w) \in \F} w$$
The following property shows the effect of negating a soft formula.
\begin{property}\label{prop-negado}
If $\F$ is a soft MaxSAT formula then 
$$\myneg \F(X)= \mathit{rf}(\F) - \F(X)$$
\end{property}

\begin{proof} 
Given a clause $(C,w)$, any truth assignment always falsifies either $(C,w)$ or $(\myneg C,w)$, but not both. Therefore, for each clause $(C,w)\in \F$, any truth assignment $X$ will incur a cost $w$ in $\F \cup \myneg \F$. Consequently, $\F(X) +\myneg  \F(X) = \sum_{(w,C)\in \F} w = \mathit{rf}(\F)$, which proves the property.

%Consider $\F(X) +\myneg  \F(X)$. Note that any truth assignment falsifies either $(C,w)$ or $(\myneg C,w)$ for every clause $(C,w)\in \F$. Therefore each truth assignment will incur a cost $w$ for every clause $(C,w)\in \F$. Consequently, $\F(X) +\myneg  \F(X) = \sum_{(w,C)\in \F} w = \mathit{rf}(\F)$, which proves the property.
\end{proof}

Next, we show that an entailment $\F \models \G$ can be rephrased as MaxSAT lower bound,

\begin{theorem} \label{TH-Maxsat-completeness}
Let $\F$ and $\G$ be two MaxSAT formulas, possibly with soft and hard clauses.
Then,
$$\F \models \G \text{ iff } MaxSAT(\F \cup {\myneg \G^{\gamma}})\geq \mathit{rf}({\G}^{\gamma})$$

\noindent where $\G^{\gamma}$ is similar to $\G$ but its infinity weights are replaced by a value $\gamma$ higher than  the maximum finite cost of $\F$,
$$\gamma > \max_{X \mid \F(X)<\infty} \F(X)$$
with $\gamma>0$ if $\F$ is unsatisfiable. 
\end{theorem}

\begin{proof}
Let us prove the if direction. $\F \models \G$ means that 
$\F(X) \geq \G(X)$ for all $X$. We know, by construction of $G^{\gamma}$ that 
$\G(X) \geq \G^{\gamma}(X)$. 
Therefore, $\F(X) \geq \G^{\gamma}(X)$ for all $X$.  Because $\G^{\gamma}$ does not contain hard clauses, $\G^{\gamma}(X)<\infty$, which means that,
$\F(X) - \G^{\gamma}(X) \geq 0$.
Adding $rf(\G^{\gamma})$ to both sides of the inequality we get,
$\F(X) + rf(\G^{\gamma}) - \G^{\gamma}(X) \geq rf(\G^{\gamma})$.
By Property~\ref{prop-negado}, we have,
$\F(X) + {\myneg \G}^{\gamma}(X) \geq rf(\G^{\gamma})$
which clearly means that,
$MaxSAT(\F \cup {\myneg \G}^{\gamma}) \geq rf(\G^{\gamma})$.

Let us prove now the only if direction.
$MaxSAT(\F \cup {\myneg \G}^{\gamma}) \geq rf({\G}^{\gamma})$ implies that $\F(X) + {\myneg \G}^{\gamma}(X) \geq rf({\G}^{\gamma})$ for all $X$. Moreover, since ${\myneg \G}^{\gamma}$ does not have hard clauses, from Property~\ref{prop-negado} we know that,
$\F(X) + rf({\G}^{\gamma}) - {\G}^{\gamma}(X) \geq rf({\G}^{\gamma})$
so we have that
$\F(X) \geq \G^{\gamma}(X)$
and we need to prove that, $\F(X) \geq \G(X)$. There are two possibilities for $\G^{\gamma}(X)$,
\begin{enumerate}
\item If $\G^{\gamma}(X) < \gamma$ it means that $X$ does not falsify any of the clauses that are hard in $\G$. Therefore, $\G^{\gamma}(X) = \G(X)$, which means that $\F(X) \geq \G(X)$.

\item If $\G^{\gamma}(X) \geq \gamma$, since $\F(X) \geq G^{\gamma}(X)$, then $\F(X) \geq \gamma$ which, by definition of $\gamma$, means that $\F(X)=\infty$. Therefore, $\F(X) \geq \G(X)$.
\end{enumerate}

\noindent which proves the theorem.
\end{proof}

\begin{example}
Consider formulas $\F=\{(z,2), (x,5), (y,\infty)\}$ and $\G=\{(x\lor z,u),(y\lor z,\infty\}$ with $u$ being a finite weight. We can apply Theorem \ref{TH-Maxsat-completeness} to find out whether $\F\models \G$. 

Clearly  $\gamma = 8 > \max_{X \mid \F(X)<\infty} \F(X)$, so we define $\G^{\gamma}=\{(x\lor z,u),(y\lor z,8)\}$, $rf(\G^{\gamma})=u+8$ and  $\myneg \G^{\gamma}=\{(\myneg x, u), (x\lor \myneg z,u), (\myneg y, 8), (y\lor \myneg z,8)\}$. With $u=5$ we have that $MaxSAT(\F \cup \myneg \G^{\gamma})=13$ and $rf(\G^{\gamma})=13$ which implies that  $\F\models \G$. However, with $u=8$ we have $MaxSAT(\F \cup \myneg \G^{\gamma})=15$ and $rf(\G^{\gamma})=16$ which implies that  $\F \nvDash \G$.

\end{example}

The following corollary will be useful in sections 6.3 and 7.2.
\begin{corollary}\label{coro-1}
A hard CNF formula $\F$ entails a hard clause $(C,\infty)$, that is $\F\models \{(C,\infty)\}$, iff $MaxSAT(\F\cup \{(\myneg C, 1)\})\geq 1$.
\end{corollary}
\begin{proof}
We can apply Theorem \ref{TH-Maxsat-completeness} with $\gamma=1$ no matter whether $\F$ is satisfiable or unsatisfiable. Then $\G^{\gamma}=\{(C,1)\}$,  $\mathit{rf}(\G^{\gamma})=1$ and $\myneg \G^{\gamma}=\{(\myneg C, 1)\}$. Hence the corollary holds.
\end{proof}

\section{MaxSAT resolution-based Proof Systems} \label{Sec-MaxRules}

MaxSAT proof systems implicitly assume the following two self-explained inference rules:

%Before that, we specify two inference rules that are assumed as implicit in any MaxSAT proof system:

\[
\begin{array}{ccc}
(C, v) \hspace{0.5cm} (C, w) & \hspace{1cm}& (C, w)\\
\cline{1-1} 
\cline{3-3} 
 (C, v+w)& & (C, v) \hspace{0.5cm} (C, w-v)\\
\texttt{(merge)} & & \texttt{(unmerge)}\\
\end{array}
\]
where in the unmerge rule $v$ must be less than $w$.
%\noindent where $w = u + v$. %In words, any clause $(C,w)$ is equivalent to two clauses $(C,u), (C,v)$ as long as the total weights are equivalent. 

In the following, we introduce and analyze the impact of three MaxSAT inference rules: \textit{resolution}, \textit{split} and \textit{virtual}. After the definition of each rule, we discuss the level of completeness that it adds to the proof system and what type of PHP problems it solves, which shows the incremental power of each proof system. % In all these proof systems, clauses are merged and unmerged as needed.

%Any clause $(C,w)$ is equivalent to two clauses $(C,u), (C,v)$ as long as $u+v=w$. In the following we will assume that clauses are unmerged and merged as needed. We say that $\G \subseteq \F$ if for all $(C,w)\in \G$ there is a $(C,w')\in \F$ with $w\leq w'$. 

%Note that if $\G \subseteq \F$ then $\F\models \G$ and there is a formula $\R$ such that $\F\equiv \G \cup \R$.

\subsection{Resolution}

The MaxSAT \textit{resolution} rule~\cite{DBLP:conf/ijcai/LarrosaH05} is

\begin{center}
\[
\begin{array}{c}
(x\lor A, w) \hspace{.5cm} (\myneg x \lor B, w) \\ 
\cline{1-1} 
(A\lor B,w)  \\
(x\lor A, w-w)\ \ \ (\myneg x\lor B, w-w)\\
(x\lor A \lor \myneg B,w)\ \ \ (\myneg x \lor B \lor \myneg A, w)\\
\end{array}
\]
\end{center}

\noindent where $A$ and $B$ are arbitrary (possibly empty) disjunctions of literals, $w> 0$. When $A$ (resp. $B$) is empty, $\myneg A$ (resp. $\myneg B$) is constant true, so $x\lor \myneg A \lor B$ (resp. $x\lor A \lor \myneg B$) is tautological. When $w\neq \infty$ the antecedents will disappear because $w-w=0$. When $w=\infty$ the antecedents are replaced by themselves because $\infty-\infty=\infty$ or, in other words, hard clauses remain throughout the proof as they do in classical SAT resolution, which means that they can be used as antecedents any number of times (see the refutation graph of Figure  \ref{fig:graphMaxSAT} (left)).

\begin{example} 
The application of MaxSAT resolution to $(x\lor y\lor z, 1)$ and $(\neg x \lor y\lor p,1)$ corresponds to,

 \[ \begin{array}{c}
(x\lor y\lor z, 1)\ \ \ (\neg x \lor y\lor p,1)  \\
\hline
(y\lor z \lor p,1)\\
(x\lor y\lor z, 0)\ \ \ (\neg x \lor y\lor p,0)\\
(x\lor y \lor z \lor \neg y,1)\ \ (\neg x \lor \neg y \lor y \lor p,1)\\
(x\lor y\lor z \lor y \lor \neg p,1)\ \ \ (\neg x \lor y \lor \neg z \lor y \lor p,1)\\
\end{array}\]

\noindent Removing zero-cost clauses, tautologies and repeated literals, the resulting set of clauses is $\{(y\lor z \lor p,1),
(x\lor y\lor z \lor \neg p,1),
(\neg x \lor y \lor \neg z \lor p,1)\}$.
\end{example}

It is known that the proof system \textbf{Res} made exclusively of the resolution rule is sound and refutationally complete \cite{DBLP:journals/ai/BonetLM07,DBLP:journals/ai/LarrosaHG08}.
However, as we show next, it is not complete.
\begin{theorem}\label{res-implicat}
Proof system \textbf{Res} is not complete.
\end{theorem}
\begin{proof} Consider formula $\F=\{(x,1), (y,1)\}$. It is clear that $\F\models (x\lor y,1)$ which cannot be derived with \textbf{Res}.
\end{proof}

 It is known that \textbf{Res} cannot compute polynomial size refutations for \textit{PHP} \cite{HAKEN1985297} or \textit{SPHP} \cite{DBLP:journals/ai/BonetLM07}. However, we show next that it can efficiently refute $SPHP^1$. We write it as a property because it will be instrumental in the proof of several results in the rest of this section. The refutation graph (which is a straightforward adaptation of what was proved in \cite{DBLP:conf/sat/IgnatievMM17} and \cite{DBLP:conf/aaai/LarrosaR20}) appears in Figure~\ref{fig:sphp1}. The refutation uses the following Lemma. 

\begin{figure}
\begin{center}
\scalebox{0.75}{\begin{tikzpicture}

\tikzset{%
every node/.style={fill=white, rounded corners, node distance=5em},%
every path/.style={->, >=stealth, line width=.06em}%
}

\tikzstyle{op}=[draw, fill=black!5!white]

\node(-x){$\myneg x_n$};

\node[op, above of=-x, yshift=-0.7cm](op-x){res};
\draw (op-x) to (-x);

\node[right of=op-x, xshift=0.5cm, yshift=0.3cm](-k1-k){$\myneg x_{n-1}\lor \myneg x_n$};
\node[above of=op-x, node distance=3em](k1-k){$x_{n-1} \lor \lnot x_n$};
\node[right of=k1-k, node distance=14em](k1-k-r){$\myneg x_{n-2} \lor \myneg x_n \lor \myneg x_{n-1} $};
\draw (-k1-k) to (op-x);
\draw (k1-k) to (op-x);

\node[op, above of=k1-k, yshift=-0.7cm](opk1-k){res};
\draw (opk1-k) to (k1-k);
\draw (opk1-k) to (k1-k-r);

\node[right of=opk1-k, xshift=0.5cm, yshift=0.3cm](-k2-k){$\myneg x_{n-2}\lor \myneg x_n$};
\node[above of=opk1-k, node distance=3em](k2-k){$x_{n-2} \lor x_{n-1} \lor \myneg x_n$};
\node[right of=k2-k, node distance=15em](k2-k-r){$\myneg x_{n-3} \lor \myneg x_n \lor \myneg{x_{n-2} \lor x_{n-1}}$};
\draw (k2-k) to (opk1-k);
\draw (-k2-k) to (opk1-k);

\node[op, above of=k2-k, yshift=-0.7cm](opk2-k){res};
\draw (opk2-k) to (k2-k);
\draw (opk2-k) to (k2-k-r);

\node[right of= opk2-k, xshift=0.5cm, yshift=0.3cm] (-k3-k) {$\myneg x_{n-3} \lor \myneg x_n$};
\node[above of=opk2-k, node distance=4.5em](3-k){$x_3 \lor \ldots \lor x_{n-1} \lor \myneg x_n$};
\node[right of=3-k, node distance=15em](3-k-r){$\myneg x_2 \lor \myneg x_n \lor \myneg{x_3 \lor \ldots \lor x_{n-1}}$};
\draw[dashed] (3-k) to (opk2-k);
\draw (-k3-k) to (opk2-k);

\node[op, above of=3-k, yshift=-0.7cm](op3-k){res};
\draw (op3-k) to (3-k);
\draw (op3-k) to (3-k-r);

\node[above of=op3-k, node distance=3em](2-k){$x_2 \lor \ldots \lor x_{n-1} \lor \myneg x_n$};
\node[right of=op3-k, xshift=0.5cm, yshift=0.3cm](-2-k){$\myneg x_{2}\lor \myneg x_n$};
\node[right of=2-k, node distance=13em](2-k-r){$\myneg x_1 \lor \myneg x_n \lor \myneg{x_2 \lor \ldots \lor x_{n-1}}$};
\node[right of=2-k-r, node distance=12em](2-k-rr){$x_1 \lor \ldots \lor x_n$};
\draw (-2-k) to (op3-k);
\draw (2-k) to (op3-k);

\node[op, above of=2-k, yshift=-0.7cm](op2-k){res};
\draw (op2-k) to (2-k);
\draw (op2-k) to (2-k-r);
\draw (op2-k) to (2-k-rr);

\node[above of=op2-k, node distance=3em](1k1){$x_1 \lor \ldots \lor x_{n-1}$};
\node[right of=1k1, node distance=7em, xshift=0.5cm](-1-k){$\myneg x_{1}\lor \myneg x_n$};
\draw (-1-k) to (op2-k);
\draw (1k1) to (op2-k);

\end{tikzpicture}}
\end{center}
\caption{Proof of Lemma \ref{le1}. All clauses have cost 1.}
\label{fig-lemma1}
\end{figure}
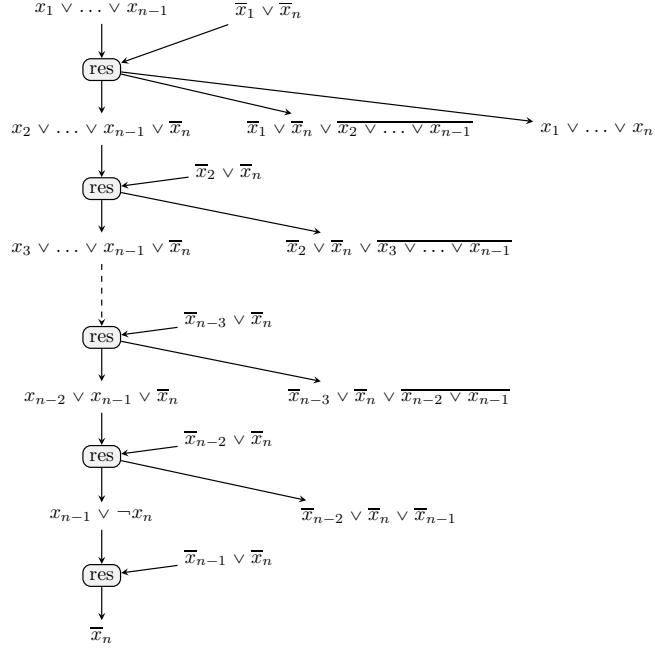

\begin{lemma} 
\label{le1}
Consider a MaxSAT formula ${\cal F} = \{(x_1\lor \ldots \lor x_{n-1}, 1)\} \cup \{(\myneg x_i \lor \myneg x_n, 1) \mid 1 \leq i < n\}$. There is a proof 
\begin{multline*}
{\cal F} \vdash_{Res} \{(\myneg x_n, 1)\} \cup \{(x_1 \lor \ldots \lor x_n, 1)\}~\cup \{(\myneg x_i \lor \myneg x_n \lor \myneg {x_{i+1} \lor \ldots \lor x_{n-1}}, 1)  \mid 1 \leq i < n-1  \}
\end{multline*}
of length $n-1$.
\end{lemma}
\begin{proof} 
The resolution proceeds as shown in Figure~\ref{fig-lemma1}.
\end{proof}

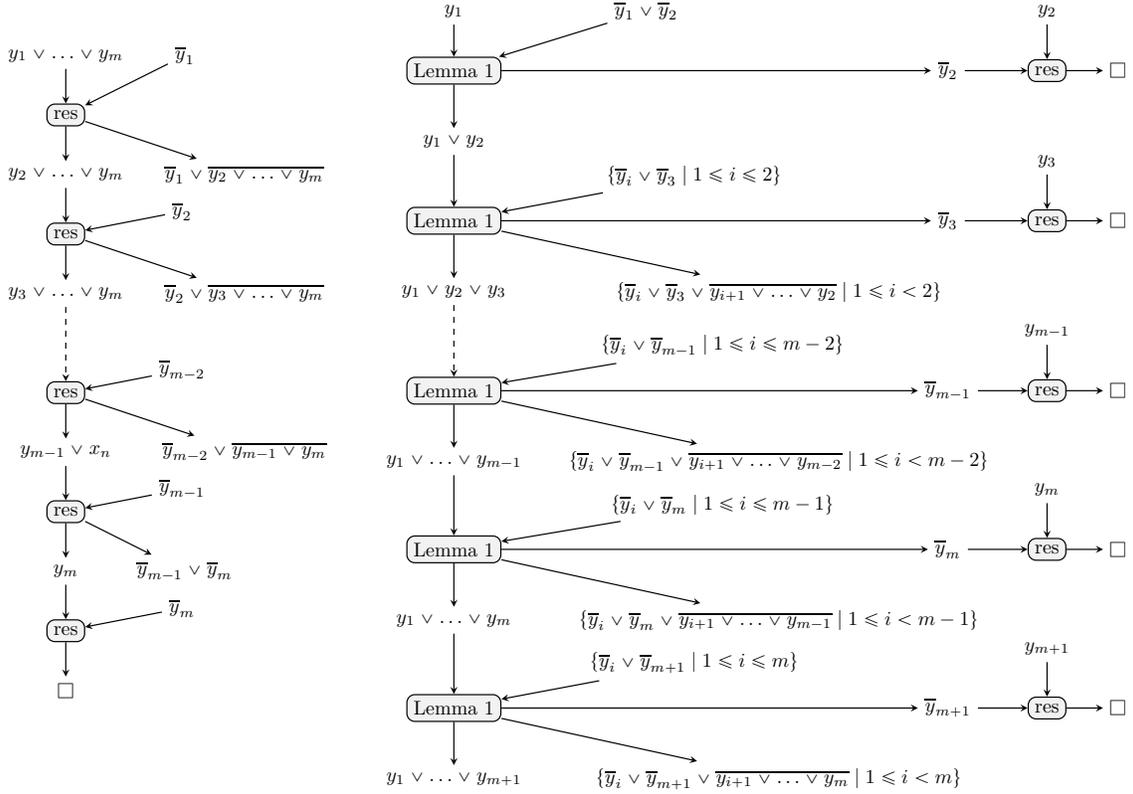
\begin{figure}[t]
\begin{center}
\scalebox{0.75}{\begin{tikzpicture}

\tikzset{%
every node/.style={fill=white, rounded corners, node distance=5em},%
every path/.style={->, >=stealth, line width=.06em}%
}

\tikzstyle{op}=[draw, fill=black!5!white]

\node(dummy){}; % nodo para centrar imagen
\node[above of=dummy](-x){$\Box$};

\node[op, above of=-x, yshift=-0.7cm](op-x){res};
\draw (op-x) to (-x);

\node[right of=op-x, xshift=0.3cm, yshift=0.4cm](-k1-k){$\myneg y_{m}$};
\node[above of=op-x, node distance=3em](k1-k){$y_{m}$};
\node[right of=k1-k, node distance=6em](k1-k-r){$\myneg y_{m-1} \lor \myneg y_m $};
\draw (-k1-k) to (op-x);
\draw (k1-k) to (op-x);

\node[op, above of=k1-k, yshift=-0.7cm](opk1-k){res};
\draw (opk1-k) to (k1-k);
\draw (opk1-k) to (k1-k-r);

\node[right of=opk1-k, xshift=0.3cm, yshift=0.4cm](-k2-k){$\myneg y_{m-1}$};
\node[above of=opk1-k, node distance=3em](k2-k){$y_{m-1} \lor x_n$};
\node[right of=k2-k, node distance=9em](k2-k-r){$\myneg y_{m-2} \lor \myneg{y_{m-1} \lor y_{m}}$};
\draw (k2-k) to (opk1-k);
\draw (-k2-k) to (opk1-k);

\node[op, above of=k2-k, yshift=-0.7cm](opk2-k){res};
\draw (opk2-k) to (k2-k);
\draw (opk2-k) to (k2-k-r);

\node[right of= opk2-k, xshift=0.3cm, yshift=0.4cm] (-k3-k) {$\myneg y_{m-2}$};
\node[above of=opk2-k, node distance=5em](3-k){$y_3 \lor \ldots \lor y_m$};
\node[right of=3-k, node distance=9em](3-k-r){$\myneg y_2 \lor \myneg{y_3 \lor \ldots \lor y_m}$};
\draw[dashed] (3-k) to (opk2-k);
\draw (-k3-k) to (opk2-k);

\node[op, above of=3-k, yshift=-0.7cm](op3-k){res};
\draw (op3-k) to (3-k);
\draw (op3-k) to (3-k-r);

\node[right of=op3-k, xshift=0.3cm, yshift=0.4cm](-2-k){$\myneg y_{2}$};
\node[above of=op3-k, node distance=3em](2-k){$y_2 \lor \ldots \lor y_m$};
\node[right of=2-k, node distance=9em](2-k-r){$\myneg y_1 \lor \myneg{y_2 \lor \ldots \lor y_m}$};
\draw (-2-k) to (op3-k);
\draw (2-k) to (op3-k);

\node[op, above of=2-k, yshift=-0.7cm](op2-k){res};
\draw (op2-k) to (2-k);
\draw (op2-k) to (2-k-r);

\node[above of=op2-k, node distance=3em](1k1){$y_1 \lor \ldots \lor y_m$};
\node[right of=1k1, node distance=6em](-1-k){$\myneg y_1$};
\draw (-1-k) to (op2-k);
\draw (1k1) to (op2-k);

\end{tikzpicture}
\hspace{0.6cm}
\begin{tikzpicture}

\tikzset{%
every node/.style={fill=white, rounded corners, node distance=5em},%
every path/.style={->, >=stealth, line width=.06em}%
}

\tikzstyle{op}=[draw, fill=black!5!white]

\node(-x){$y_1\lor \ldots \lor y_{m+1}$};
\node[right of=-x,xshift=4cm](-nr){$\{ \myneg y_i \lor \myneg y_{m+1} \lor \myneg{y_{i+1} \lor \ldots \lor y_{m}} \mid 1 \leq i < m\} $};

\node[op, above of=-x, yshift=-0.5cm](op-x){Lemma~\ref{le1}};
\node[right of=op-x,xshift=7cm](-n){$\myneg y_{m+1} $};
% --
\node[op, right of=-n](resxn){res};
\node[above of=resxn, yshift=-0.7cm](xn){$y_{m+1}$};
\draw (xn) to (resxn);
\draw (-n) to (resxn);
\node[right of=resxn, xshift=-0.5cm](boxxn){$\Box$};
\draw (resxn) to (boxxn);
% ----

\draw (op-x) to (-x);
\draw (op-x) to (-n);
\draw (op-x) to (-nr);

\node[right of=op-x, xshift=2.5cm, yshift=0.8cm](-k1-k){$\{ \myneg y_i \lor \myneg y_{m+1} \mid 1 \leq i \leq m\}$};
\node[above of=op-x, yshift=-0.2cm](k1-k){$y_1 \lor \ldots \lor y_{m}$};
\node[right of=k1-k,xshift=4cm](k1-k-rr){$\{ \myneg y_i \lor \myneg y_{m} \lor \myneg{y_{i+1} \lor \ldots \lor y_{m-1}} \mid 1 \leq i < m- 1\} $};
\draw (-k1-k) to (op-x);
\draw (k1-k) to (op-x);

\node[op, above of=k1-k, yshift=-0.5cm](opk1-k-r){Lemma~\ref{le1}};
\node[right of=opk1-k-r,xshift=7cm](k1-k-r){$\myneg y_{m} $};
% --
\node[op, right of=k1-k-r](resxn-1){res};
\node[above of=resxn-1, yshift=-0.7cm](xn-1){$y_{m}$};
\draw (xn-1) to (resxn-1);
\draw (k1-k-r) to (resxn-1);
\node[right of=resxn-1, xshift=-0.5cm](boxxn-1){$\Box$};
\draw (resxn-1) to (boxxn-1);
% ----
\draw (opk1-k-r) to (k1-k);
\draw (opk1-k-r) to (k1-k-r);
\draw (opk1-k-r) to (k1-k-rr);

\node[right of=opk1-k-r, xshift=3cm, yshift=0.8cm](-k2-k){$\{ \myneg y_i \lor \myneg y_{m} \mid 1 \leq i \leq m - 1\}$};
\node[above of=opk1-k-r, yshift=-0.2cm](k2-k){$y_1 \lor \ldots \lor y_{m-1}$};
\node[right of=k2-k,xshift=4cm](k2-k-rr){$\{ \myneg y_i \lor \myneg y_{m-1} \lor \myneg{y_{i+1} \lor \ldots \lor y_{m-2}} \mid 1 \leq i < m- 2\} $};
\draw (k2-k) to (opk1-k-r);
\draw (-k2-k) to (opk1-k-r);

\node[op, above of=k2-k, yshift=-0.5cm](opk2-k){Lemma~\ref{le1}};
\node[right of=opk2-k, xshift=7cm](k2-k-r){$\myneg y_{m-1}$};
% --
\node[op, right of=k2-k-r](resxn-2){res};
\node[above of=resxn-2, yshift=-0.7cm](xn-2){$y_{m-1}$};
\draw (xn-2) to (resxn-2);
\draw (k2-k-r) to (resxn-2);
\node[right of=resxn-2, xshift=-0.5cm](boxxn-2){$\Box$};
\draw (resxn-2) to (boxxn-2);
% ----
\draw (opk2-k) to (k2-k);
\draw (opk2-k) to (k2-k-r);
\draw (opk2-k) to (k2-k-rr);

\node[right of= opk2-k, xshift=3cm, yshift=0.8cm] (-k3-k) {$\{ \myneg y_i \lor \myneg y_{m-1} \mid 1 \leq i \leq m-2\}$};
\node[above of=opk2-k](3-k){$y_1 \lor y_2 \lor y_3$};
\node[right of=3-k,xshift=4cm](3-k-rr){$\{ \myneg y_i \lor \myneg y_{3} \lor \myneg{y_{i+1} \lor \ldots \lor y_{2}} \mid 1 \leq i < 2\} $};
\draw[dashed] (3-k) to (opk2-k);
\draw (-k3-k) to (opk2-k);

\node[op, above of=3-k, yshift=-0.5cm](op3-k){Lemma~\ref{le1}};
\node[right of=op3-k, xshift=7cm](3-k-r){$\myneg y_3$};
% --
\node[op, right of=3-k-r](resx3){res};
\node[above of=resx3, yshift=-0.7cm](x3){$y_3$};
\draw (x3) to (resx3);
\draw (3-k-r) to (resx3);
\node[right of=resx3, xshift=-0.5cm](boxx3){$\Box$};
\draw (resx3) to (boxx3);
% ----
\draw (op3-k) to (3-k);
\draw (op3-k) to (3-k-r);
\draw (op3-k) to (3-k-rr);

\node[right of=op3-k, xshift=2.5cm, yshift=0.8cm](-2-k){$\{ \myneg y_i \lor \myneg y_3 \mid 1 \leq i \leq 2\}$};
\node[above of=op3-k, node distance=4em](2-k){$y_1 \lor y_2$};
\draw (-2-k) to (op3-k);
\draw (2-k) to (op3-k);

\node[op, above of=2-k, yshift=-0.5cm](op2-k){Lemma~\ref{le1}};
\node[right of=op2-k, xshift=7cm](2-k-r){$\myneg y_2$};
% --
\node[op, right of=2-k-r](resx2){res};
\node[above of=resx2, yshift=-0.7cm](x2){$y_2$};
\draw (x2) to (resx2);
\draw (2-k-r) to (resx2);
\node[right of=resx2, xshift=-0.5cm](boxx2){$\Box$};
\draw (resx2) to (boxx2);
% ----
\draw (op2-k) to (2-k);
\draw (op2-k) to (2-k-r);

\node[above of=op2-k, node distance=3em](1k1){$y_{1}$};
\node[right of=1k1, node distance=4em, xshift=2cm](-1-k){$\myneg y_1 \lor \myneg y_2$};
\draw (-1-k) to (op2-k);
\draw (1k1) to (op2-k);

\end{tikzpicture}}
\end{center}
\caption{Left: derivation graph corresponding to pigeon $i$ (for clarity purposes, we rename each $x_{ij}$, $1 \leq j \leq m$, to $y_j$). Right: derivation graph corresponding to hole $j$ (for clarity purposes, we rename each variable $x_{ij}$, $1\leq i \leq m + 1$, to $y_i$.). All clauses have cost 1.}
\label{fig:sphp1}
\end{figure}
% ------

 \begin{property}
 \label{prop-Res-sphp1}
There is a polynomial size \textbf{Res} refutation of $SPHP^1$.
 \end{property}{}
\begin{proof}
The refutation is divided in two parts. First, for each one of the $m+1$ pigeons there is a derivation
$$\{(x_{i1}\lor x_{i2}\lor \ldots \lor x_{im},1)\} \cup \{(\myneg x_{ij},1)|\ 1\leq j \leq m\}\vdash_{Res} (\Box, 1)$$
Figure~\ref{fig:sphp1} (left) shows the derivation graph that corresponds to an arbitrary pigeon $i$.
Second, for each one of the $m$ holes there is a derivation 
$$\{(\myneg x_{ij} \lor \myneg x_{i'j},1) \mid 1\leq i<i'\leq m+1\} \cup \{(x_{ij},1) \mid 1\leq i \leq m+1\}\vdash_{Res} \{(\Box, m)\}$$
Figure~\ref{fig:sphp1} (right) shows the derivation graph that corresponds to an arbitrary hole $j$.

Because each derivation is independent of the other, they can be done one after another, aggregating all the empty clauses, which produces 
$$\mathit{SPHP}^1 \vdash_{Res} \{(\Box,m^2+m+1)\}$$
which is a refutation of $\mathit{SPHP}^1$.
Observe that each pigeon proof has length $O(m)$ and each hole proof has length $O(m^2)$. Therefore, the length of the refutation is $O(m^3)$.
\end{proof}

\begin{property}
There is no polynomial size \textbf{Res} refutation of $SPHP^0$.
\end{property}
\begin{proof}
\textbf{Res} cannot produce a polynomial size refutation for $SPHP^0$ because the resolution rule cannot be applied to the empty clause $(\Box, w)$, so it must remain unaltered during any derivation. If \textbf{Res} could refute $\mathit{SPHP^0}$ in polynomial time it would also refute $\mathit{SPHP}$ in polynomial time, which is not the case~\cite{DBLP:journals/ai/BonetLM07}.
\end{proof}

\subsection{Split}

The \textit{split} rule,
\begin{center}
\[
\begin{array}{c}
(A, w) \\ 
\cline{1-1} 
(A\lor x,w) \hspace{.5cm} (A\lor \myneg x, w)
\end{array}
\]
\end{center}

\noindent is the natural extension of its SAT counterpart. 

\begin{theorem}
The split rule is sound.
\end{theorem}
\begin{proof}
We have to prove that $\F \cup \{(A,w)\} \equiv \F \cup \{(A\lor x,w), (A\lor \myneg x, w)\}$
Consider an arbitrary truth assignment.
If it satisfies $A$, then it also satisfies $A\lor x$ and $A\lor \myneg x$ so the the cost of the truth assignment is the same before and after the split. If the truth assignment does not satisfy $A$, then there is a cost of $w$ caused by $A$. After the application of the split the same cost will be caused either by  $A \lor x$ or by $A \lor \myneg x$ depending on whether the truth assignment satisfies $x$ or not.
\end{proof} 

The proof system \textbf{ResS}, made of resolution and split,
is sound and complete.

\begin{theorem}
Proof system \textbf{ResS} is sound.
\end{theorem}
\begin{proof}
We have to prove that $\F \vdash_{ResS} \G$ implies $\F \models \G$. Because resolution and split are sound, $\F \vdash_{ResS} \G$ implies that there is a derivation $\F \vdash_{ResS} \G \cup \R$ for some $\R$ such that $\forall X, \F(X) = \G(X)+\R(X)$. Therefore, $\F \models \G$, which completes the proof.
\end{proof}

\begin{theorem}
Proof system \textbf{ResS} is complete.
\end{theorem}
\begin{proof}
We have to prove that if $\F \models \G$ then there is derivation $\F \vdash_{ResS} \G$.
The proof is based on the following two facts:
\begin{enumerate}
    \item For every formula $\F$ there is a derivation $\F\vdash_{ResS} \F^{ext}$ made exclusively of splits and merges such that: $i$) $\F \equiv \F^{ext}$, $ii$) all the clauses of $\F^{ext}$ contain all the variables in the formula and $iii$) there are no repeated clauses. In the derivation each clause $(C, w) \in \F$ can be expanded to a new variable not in $C$ using the split rule. The process is repeated until all clauses in the current formula contain all the variables in the formula. Then, pairs of equal clauses $(C', u)$, $(C', v)$ are merged and, thus, $\F^{ext}$ does not contain repeated clauses. As a result,  $\F^{ext}$ contains one clause $(C,w)$ for each $\F^{ext}(X) = w>0$, where $C$ is falsified exactly by $X$.
    \item If there is a derivation $\F \vdash_{ResS} \F^{ext}$ made exclusively of splits and merges, then there is a derivation $\F^{ext} \vdash_{ResS} \F$ made exclusively of resolutions and unmerges. Let $\Pi=(\F_0 = \F; \F_1; \ldots; \F_p = \F_{ext})$ be the first derivation. Then, the later derivation is $\Gamma = (\F_{ext} = \F_p; \F_{p-1}; \ldots; \F_0)$ where $\F_{i}; \F_{i-1} \in \Gamma$ is an unmerge if $\F_{i- 1}; \F_{i} \in \Pi$ is a merge; and $\F_{i}; \F_{i-1} \in \Gamma$ is a resolution if $\F_{i-1}; \F_{i} \in \Pi$ is an split. 
    
    %The later derivation is similar to the former one, but in reverse order and it applies resolution to the pairs of clauses that are split in the first place.

\end{enumerate}

From fact (1) we know that $\F \vdash_{ResS} \F^{ext}$. 
Since $\F \models \G$ we know $\F^{ext} \models \G^{ext}$. We can separate $\F^{ext}$ as $\F^{ext}=\G^{ext} \cup \R$. From fact (1) and (2) we know that $\G^{ext} \vdash_{ResS} \G$. Joining the two derivations we have $\F \vdash_{ResS} \F^{ext} \vdash_{ResS} \R \cup \G$, which proves the theorem.

\end{proof}

We show now which pigeon problems \textbf{ResS} can and cannot solve.

\begin{property}
There is a polynomial size \textbf{ResS} refutation for $SPHP^0$.
\end{property}
\begin{proof}
\textbf{ResS} can produce a polynomial size refutation for $SPHP^0$ because it can transform $SPHP^0$ into $SPHP^1$ and then apply Property~\ref{prop-Res-sphp1}. The transformation is done by a sequence of splits,
\begin{center}
\[
\begin{array}{c}
(\Box, 1) \\ 
\cline{1-1} 
(x_{ij}, 1) \hspace{.5cm} (\myneg x_{ij}, 1)
\end{array}
\]
\end{center}

\noindent that move one unit of weight from the  empty clause to every variable in the formula and its negation.
\end{proof}

\begin{property}\label{propabc}
There is no polynomial size \textbf{ResS} refutation  for $PHP$.
\end{property}
\begin{proof}
ResS with hard formulas corresponds to the SAT proof system containing SAT resolution and SAT split. From Property~\ref{sat-res-ress}, we know that it is equivalent to the SAT proof system containing only resolution.
Therefore, the existence of a polynomial  size \textbf{ResS} refutation  for $PHP$ would imply the existence of a polynomial size refutation with SAT resolution, which is not possible \cite{Robinson:1965:MLB:321250.321253}.
\end{proof}

\begin{property}
There is no polynomial size \textbf{ResS} refutation for $SPHP$.
\end{property}
\begin{proof}
We show that we can build a  \textbf{ResS} refutation for $PHP$ from a  \textbf{ResS} refutation for $SPHP$ without increasing its length. Therefore, a polynomial size refutation for $SPHP$ would imply a polynomial size refutation for $PHP$, which is a contradiction to Property~\ref{propabc}.

Let $\Pi=(\F_0; \F_1; \ldots; \F_e)$ with $SPHP=\F_0$ and $(\Box,1)\in \F_e$ be the refutation and $G(\Pi)$ its associated graph. We are going to transform $G(\Pi)$ into a $PHP$ refutation following the derivation steps. First, replace weight $1$ by $\infty$ in all the zero in-neighbors clauses (namely, original clauses). Then follow the refutation step by step. If the inference step is a split, just replace the weight of the consequents by infinity. If the inference is a resolution between $x\lor A$ and $\myneg x \lor B$, merge nodes $\{A\lor B, x\lor A \lor \myneg B, \myneg x \lor \myneg A \lor B\}$ into $A\lor B$, and replace the weight of all the consequents by infinity. By construction, when considering any inference step all its in-neighbors will already have infinity weight making the graph correct. At the last step, node $(\Box,1)$ will be transformed into $(\Box,\infty)$ making the graph a $PHP$ refutation.
\end{proof}

A consequence of the previous results is that, unlike what happens in the SAT case (see Property \ref{sat-res-ress}), \textbf{ResS} is stronger than \textbf{Res},
\begin{theorem}
\textbf{ResS} is stronger than \textbf{Res}.
\end{theorem}{}
\begin{proof}
On the one hand, it is clear that \textbf{ResS} can $p$-simulate any proof of \textbf{Res} since it is a superset of \textbf{Res}. On the other hand \textbf{Res} cannot $p$-simulate \textbf{ResS} because there is a polynomial size \textbf{ResS} refutation of $SPHP^0$ which cannot exist for \textbf{Res}.
\end{proof}

\noindent
Next we show that, similarly to what happens in the SAT case, the split rule allows to restrict the use of resolution to its symmetric form (this result will be useful in Section~ \ref{Sec-Circular}).
The \textit{symmetric resolution} rule,
\begin{center}
\[
\begin{array}{c}
(A\lor x,w) \hspace{.5cm} (A\lor \myneg x, w) \\ 
\cline{1-1} 
(A,w)\\
(A\lor x,w - w) \hspace{.5cm} (A\lor \myneg x, w - w) \\ 
\end{array}
\]
\end{center}

\noindent is the natural extension of its SAT counterpart. In combination with split, symmetric resolution already guarantees completeness.

\begin{property}\label{prop-symres}
The MaxSAT resolution rule can be replaced by $O(n)$ splits and one symmetric resolution, where $n$ is the number of variables in the formula.

%MaxSAT resolution and symmetric MaxSAT resolution are equivalent in \textbf{ResS}. 

%If there is a proof $\F \vdash_{ResS} \G$ of length $e$, then there is a ResS proof $\F \vdash_{ResS} \G$ of length less than $n \cdot e$ ($n$ being the number of variables) in which resolution is restricted to its symmetric form.
\end{property}

\begin{proof}
 Consider clauses $(x \lor A, u)$ and $(\myneg x \lor B, u)$. $|C-A|$ splits transform the first clause into  $\{(x \lor A \lor B, u), (x \lor A \lor \myneg B, u)\}$. Similarly, $|C-B|$ splits transform the second clause into  $\{(\myneg x \lor A \lor B, u), (\myneg x \lor B \lor \myneg A, u)\}$. Finally, it is possible to apply symmetric resolution between $(x \lor A \lor B, u)$ and $(\myneg x \lor A \lor B, u)$, which proves our claim.

%Consider clauses $(x \lor A, u)$ and $(\myneg x \lor B, v)$. Let $m=\min\{u,v\}$ and $C=A\lor B$. One unmerge and $|C-A|$ splits transform the first clause into  $\{(x \lor A, u-m), (x \lor A \lor B, u-m), (x \lor A \lor \myneg B, u-m)\}$. Similarly, one unmerge and $|C-B|$ splits transform the second clause into  $\{(x \lor B, v-m), (x \lor A \lor B, v-m), (x \lor B \lor \myneg A, v-m)\}$. Now it is possible to apply symmetric resolution between $(x \lor A \lor B, u-m)$ and $(x \lor A \lor B, v-m)$, which proves our claim.
\end{proof}

\subsection{Virtual}
Now we introduce our third and last rule, \textit{virtual}, and show that it can further speed-up refutations. Roughly speaking, it allows to anticipate weighted clauses that will be derived later on and use them right away. Any derivation obtained from this anticipated clauses will be sound as long as the anticipation turns out to be true.
The \textit{virtual} rule is,

\begin{center}
\[\begin{array}{c}
\\
\hline
(A,w) \hspace{.5cm} (A,-w)
\end{array}
\]
\end{center}
with $w\neq \infty$. It allows to introduce a fresh clause $(A,w)$ into the formula. To preserve soundness (i.e, cancel out the effect of the addition) it also adds $(A,-w)$. The use of virtual requires to allow clauses with negative finite weights \footnote{Note that the virtual rule can be seen as a generalization of the unmerge rule. Here we prefer to define it as an independent rule for clarity purposes.}.

\begin{theorem}
The virtual inference rule is sound.
\end{theorem}
\begin{proof}
We have to prove that the cost of any truth assignment is the same for $\F$ and $\F \cup \{(A,w),(A,-w)\}$.
If the truth assignment satisfies $A$, then the new clauses are also satisfied and they do not affect its cost. If the truth assignment does not satisfy $A$, the cost will be increased by $w$ because of the first clause and decreased by $w$ because of the second clause, which leaves the total cost unaltered.
\end{proof}

Let \textbf{ResSV} be the proof system made of resolution, split and virtual. Recall that resolution and split were only defined for antecedents with positive weights and we keep this restriction in the \textbf{ResSV} proof system. Therefore, they can use as an antecedent positive clauses introduced by virtual, but not the negative clauses. 

The following theorem indicates that proof system \textbf{ResSV} is sound, but the definition of soundness requires a technical redefinition of $\vdash$.
In Section \ref{Sec-MaxPS} we introduced $\F\vdash_{S} \G$ to denote an arbitrary proof $\F;\F_1;\ldots;\R$ with $\G\subseteq \R$ under proof system $\textbf{S}$, and defined the soundness of $\textbf{S}$ using that notation. Because the virtual rule introduces negative weights, this definition needs to be revised. To see why, consider a one step derivation $\{\};\{(\Box,-1), (\Box,1\}$ that only applies the virtual rule. Clearly, $\{\}\vdash_{ResSV} (\Box,1)$. However, $\{\}$ corresponds to constant zero and $(\Box,1)$ corresponds to constant $1$ and it is false that $0\geq 1$ (i.e.,  $\{\} \nvDash (\Box, 1)$). We solve this problem by redefining $\vdash_S$. 

\begin{definition}[$\vdash$]
$\F\vdash_{S} \G$ denotes an arbitrary $\textbf{S}$ proof $\F;\ldots;\R$ with $\G\subseteq \R$ \textit{and all the clauses in} $\R$ \textit{having positive weights}
\end{definition}

Note that this new definition does not affect proof systems \textbf{Res} and \textbf{ResS} because they always deal with positive weights.

\begin{theorem}
Proof system \textbf{ResSV} is sound.
\end{theorem}
\begin{proof}
We have to prove that $\F\vdash_{ResSV} \G$ implies $\F \models \G$.
Consider an arbitrary derivation $\F\vdash_{ResSV} \G$.
By definition of $\vdash_{ResSV}$, $\F\vdash_{ResSV} \G$ means $\F; \ldots; {\cal H}$ where $\G \subseteq {\cal H}$ and all clauses in ${\cal H}$ have positive weight. 
Because resolution, split and virtual are sound, we have that $\F(X)= \G(X) + {\cal R}(X)$, where ${\cal H} = \G \cup {\cal R}$. Therefore $\F(X)\geq \G(X)$, which completes the proof. 
\end{proof}

Figure \ref{fig:graphMaxSAT} (right) shows a refutation graph with \textbf{ResSV}. Note that the refutation is correct since all nodes with no out-neighbours have positive weight.

The intuition behind the virtual rule and its soundness theorem is that the rule introduces hypothetical clauses that can be temporarily used to derive new knowledge, but this new knowledge is valid only if the proof manages to cancel out the clauses with negative weight. Since negative clauses cannot be manipulated by inference rules, one way to interpret them is like a \textit{reminder} of what needs to be re-derived to make the proof sound.

Next, we discuss the completeness of \textbf{ResSV}. Note that completeness of \textbf{ResSV} is obvious since \textbf{ResS} is complete, so we can just ignore the virtual rule in any \textbf{ResSV} proof. However, a related and more interesting question is whether the use of the virtual rule can take an ongoing proof to a state from which the objective formula cannot be derived. If that was the case, the practical use of \textbf{ResSV} would be jeopardized. The following theorem shows that this is not the case. No matter which are the first inference steps, we can always proceed with the derivation, get rid of the negative clauses introduced by the virtual rule, and end up deriving any entailed formula. To prove that, we find useful the following lemma.

\begin{lemma}
 \label{l}
 There is a \textbf{ResSV} proof ${\cal F};\ldots; {\cal F} \cup \{(\Box, -w), (C,w),(\myneg C, w)\}$ for any formula ${\cal F}$, clause $C$ and weight $0<w$.
\end{lemma}
\begin{proof}
Let $C=l_1\lor l_2 \lor \cdots \lor l_r$.
The derivation is done by first introducing $(\Box, w)$ and $(\Box, -w)$ with the virtual rule, followed by a sequence of $r$ splits,
\[
\begin{array}{l}
{\cal F}; {\cal F}\cup \{(\Box,-w),(\Box,w)\};\\
{\cal F}\cup \{(\Box,-w),(\myneg l_1,w),(l_1,w)\};\\
{\cal F}\cup \{(\Box,-w),(\myneg l_1,w),(l_1\lor \myneg l_2,w),(l_1\lor l_2,w)\};\\
;\ldots;\\
{\cal F}\cup \{(\Box,-w),(\myneg l_1,w),(l_1\lor \myneg l_2,w),\ldots,\\
(l_1\lor l_2\lor \ldots \lor \myneg l_r ,w), (l_1\lor l_2\lor \ldots \lor l_r ,w)\}
\end{array}
\]

\noindent By Property~\ref{propNegCl}, the last element in the derivation is equivalent to
$$ {\cal F}\cup \{(\Box,-w),(\myneg C,w),(C, w)\}$$

\end{proof}

\begin{theorem}
Consider formulas ${\cal F}$ and $\G$ such that $\F\models \G$, and a \textbf{ResSV} proof ${\cal F};\F_1;\ldots;{\cal F}_i$. 
There is a proof  ${\cal F}_i \vdash_{ResSV}  \G$.
\end{theorem}
\begin{proof}
Let ${\cal N}\subseteq {\cal F}_i$ be the set of clauses with negative weights. If ${\cal N} = \emptyset$ then completeness follows trivially from the completeness of \textbf{ResS}. Otherwise, for each $(C,-w) \in {\cal N}$ we add (using the previous lemma) $\{(\Box, -w),(C,w),(\myneg C,w)\}$. After clause merging, ${\cal F}_i; \ldots; \{(\Box,-r)\}\cup {\cal F}'_j$ 
%the resulting formula is 
%${\cal F}_j=\{(\Box,-r)\}\cup {\cal F}'_j$ 
with $-r=\sum_{(C,-w)\in {\cal N}}-w$ being a negative number and ${\cal F}'_j$ contains only positive weights because each $(C,-w)$ vanishes when aggregating $(C,w)$.

Since the three inference rules in \textbf{ResSV} are sound, we have that
$\F(X) = -r + {\cal F}'_j(X)$, which implies that $\F'_j(X)\geq\F(X)$. Together with $\F(X) \geq
\G(X)$, they imply $\F'_j(X) \geq \G(X) + r$, which means that $\F'_j
\models \G \cup \{(\Box, r)\}$. Since \textbf{ResS} is complete, 
$\F'_j
\vdash_{ResS} \G \cup \{(\Box, r)\}$ (i.e., $\F_j'; \ldots; \G \cup \{(\Box, r)\} \cup {\cal H}$ where all clauses in $\G \cup {\cal H}$ have positive weights). Joining the two
facts,

$$\F_i; \ldots; \{(\Box, -r)\} \cup \F'_j; \ldots; \{(\Box, -r)\} \cup  \G \cup \{(\Box, r)\} \cup {\cal H}$$
After merging $(\Box, -r)$ and $(\Box, r)$, the previous derivation can be written as $\F_i \vdash_{ResSV} \G$.
\end{proof}

\begin{property}
There is a polynomial size \textbf{ResSV} refutation of $SPHP$.
\end{property}
\begin{proof}
First, for each variable $x_{ij}$ in $SPHP$ we introduce clauses $(x_{ij}, 1)$, $(x_{ij}, -1)$ and $(\myneg x_{ij}, 1)$, $(\myneg x_{ij}, -1)$ thanks to the virtual rule. As a consequence, $$SPHP; \ldots; SPHP^1 \cup \{(x_{ij}, -1), (\myneg x_{ij}, -1) \mid 1 \leq i \leq m + 1, 1 \leq j \leq m\}$$

\noindent Since, by Property~\ref{prop-Res-sphp1}, there exists a proof $SPHP^1; \ldots; \G \cup \{(\Box, m^2+m+1)\}$ where all clauses in $\G$ have positive weights, then
$$SPHP ; \ldots; \G \cup \{(\Box, m^2+m+1)\} \cup \{(x_{ij}, -1), (\myneg x_{ij}, -1) \mid 1 \leq i \leq m + 1, 1 \leq j \leq m\}$$

%introduce $m^2 + m + 1$
\noindent Finally, clause $(\Box, m^2+m+1)$ is unmerged to $m^2 + m + 1$ clauses $(\Box, 1)$ and each $(\Box, 1)$ is split to one pair $(x_{ij}, 1), (\myneg x_{ij}, 1)$. Since there are $m^2 + m$ variables, one clause $(\Box, 1)$ still remains. That is,

$$SPHP; \ldots; \G \cup \{(\Box, 1)\} \cup \{(x_{ij}, 1), (\myneg x_{ij}, 1) \mid 1 \leq i \leq m + 1, 1 \leq j \leq m\} \cup $$ $$\hspace{4cm} \cup \{(x_{ij}, -1), (\myneg x_{ij}, -1) \mid 1 \leq i \leq m + 1, 1 \leq j \leq m\}$$

\noindent After merging clauses with positive and negative weights,

$$SPHP \vdash_{ResSV} \{(\Box, 1)\}$$

\noindent The length of the refutation is $O(m^3)$.

%The short refutation of \textbf{ResSV} is obtained by first virtually transforming $SPHP$ into $SPHP^1$. Then, it uses Property~\ref{prop-Res-sphp1} to derive $(\Box, m^2+m+1)$. Finally, it splits one unit of the empty clause cost to each pair $x_{ij}, \myneg x_{ij}$ to cancel out negative weights. At the end of the process all clauses have positive weight while still having $(\Box,1)$.
\end{proof}

The main consequence of the previous property is that \textbf{ResSV} is stronger than \textbf{ResS},
\begin{theorem}
 \textbf{ResSV} is stronger than \textbf{ResS}.
\end{theorem}{}
\begin{proof}
On the one hand, it is clear that \textbf{ResSV} $p$-simulates \textbf{ResS} since it is a superset of \textbf{ResS}. On the other hand, \textbf{ResSV} is can produce a polynomial size refutation of $SPHP$, while \textbf{ResS} cannot.
\end{proof}

We will finish this section showing that Theorem~\ref{TH-Maxsat-completeness} has an unexpected application in the context of \textbf{ResSV}. Consider the $PHP$ problem. In MaxSAT, proving its unsatisfiability means proving $MaxSAT(PHP)=\infty$. This can be done with a refutation  $PHP\vdash (\Box,\infty)$, or using Corollary~\ref{coro-1}, which tells that $\F\models (\Box,\infty)$ if and only if $MaxSAT(\F)\geq 1$, which corresponds to a weaker derivation $PHP\vdash (\Box,1)$. The following two theorems shows that \textbf{ResSV} cannot do efficiently the first approach, but can do efficiently the second.
\begin{theorem}
There is no polynomial size proof $PHP\vdash_{ResSV} (\Box,\infty)$.
\end{theorem}
\begin{proof}
By definition, the virtual rule cannot introduce hard clauses. Resolution and split only produce new hard consequents if their antecedents are hard. Therefore, $(\Box, \infty)$ can only be obtained by resolving or splitting hard clauses in $PHP$. Consequently, if there is a polynomial size refutation for $PHP\vdash_{ResSV} (\Box,\infty)$, then it is a polynomial size \textbf{ResS} refutation $PHP\vdash_{ResS} (\Box,\infty)$. Property \ref{sat-res-ress} tells that it would also imply the existence of a polynomial size \textbf{Res} refutation $PHP\vdash_{ResS} (\Box,\infty)$ which is impossible.
\end{proof}

\begin{theorem}\label{th-phphard}
There is a polynomial size \textbf{ResSV} proof of $(\Box,1)$ from $PHP$.
\end{theorem}
\begin{proof}
We only need to apply the virtual rule,
\begin{center}
\[\begin{array}{c}
 \\ 
\cline{1-1} 
(\Box,m^2+m) \hspace{.5cm} (\Box,-m^2-m)
\end{array}
\]
\end{center}
and then split,
\begin{center}
\[
\begin{array}{c}
(\Box, 1) \\ 
\cline{1-1} 
(x_{ij}, 1) \hspace{.5cm} (\myneg x_{ij}, 1)
\end{array}
\]
\end{center}
for each $i,j$. Then, we can unmerge each (hard) clause of $PHP$ extracting weight one.  The resulting problem is $PHP \cup PHP^1$. At this point  
the proof of Property~\ref{prop-Res-sphp1} shows that we can derive $(\Box,m^2+m+1)$ which cancels out the negative weight while still retaining $(\Box,1)$.
\end{proof}

\section{ResSV and Circular Proofs}
\label{Sec-Circular}

In this section we study the relation between \textbf{ResSV} and the recently proposed concept of circular proofs  \cite{DBLP:conf/sat/AtseriasL19}. Circular proofs allow the addition of an arbitrary set of clauses to the original formula. It can be seen that conclusions are sound as long as the added clauses are \textit{re-derived} as many times as they are used. This condition is characterized as the existence of a flow in a graphical representation of the proof.
Since Circular proofs are defined in the context of hard formulas, the comparison has to be circumscribed to that context. Here we show that the \textbf{ResSV} proof system naturally captures the same idea with an arguably simpler notation. In particular, the virtual rule with its soundness theorem that requires that weights must be positive at the end of the derivation guarantees the existence of the flow.

\subsection{Circular Proofs}

 Given a CNF formula $\F$ and a SAT proof system $\textbf{S}$, a \textit{circular pre-proof} of $C_r$ from $\F$ is a SAT proof
$$\Pi= (C_1,C_2,\ldots,C_p,C_{p+1}, C_{p+2},\ldots, C_{p+q}, C_{p+q+1}, C_{p+q+2},\ldots,C_r)$$
such that 
$\F=\{C_1,C_2,\ldots,C_p\}$,
${\cal B}=\{C_{p+1}, C_{p+2},\ldots, C_{p+q}\}$ 
is an arbitrary set of clauses, and each $C_{i}$ ( with $i>p+q$) is obtained from previous clauses by applying an inference rule in $\textbf{S}$. Therefore, a pre-proof is no more than a proof where the original formula $\F$ is augmented with an arbitrary set of new clauses ${\cal B}$. 

 A \textit{circular pre-proof} $\Pi$ is associated with a (possibly cyclic) directed bi-partite graph $G(\Pi)$. To define such graph, consider first the acyclic graph as defined in Section~\ref{Sec-SATPS} using $\F \cup {\cal B}$ as the start of the proof. $G(\Pi)$ is the compactation of that graph by considering every clause in $C\in \B$ and merging all nodes whose associated clause is identical to it. After the compactation the graph may become cyclic due to the back-edges from derived clauses that were already in ${\cal B}$. 

A \textit{flow assignment} for a circular pre-proof is an assignment $f:I\longrightarrow \mathbb{N} \setminus \{0\}$ of inference nodes to positive integers (see Lemma 1 in~\cite{DBLP:conf/sat/AtseriasL19}).  The \textit{balance} of node $C\in J$ is the inflow minus the outflow,
$$b(C)= \sum_{R \in N^-(C)} f(R) - \sum_{R \in N^+(C)} f(R)$$
where $N^-(C)$ and $N^+(C)$ denote the set of in and out-neighbors of node $C\in J$, respectively. 

\begin{definition}
Given a SAT proof system $\textbf{S}$, a SAT circular proof under $\textbf{S}$ of clause $A$ from CNF formula $\F$ is a pre-proof $\Pi$ whose proof-graph $G(\Pi)$ admits a flow in which all clauses not in $\F$ have non-negative balance and $A$ has a strictly positive balance. 
\end{definition}

\begin{property} (see Section 3.5 in~\cite{DBLP:journals/corr/abs-1802-05266})
An inference rule satisfies the \emph{multiple consequence property} iff any truth assignment that falsifies one of its consequent formulas satisfies all other consequent formulas. 
\end{property}

\begin{theorem} (see Theorem 4 in~\cite{DBLP:journals/corr/abs-1802-05266})
Assuming a sound SAT proof system $\textbf{S}$ such that all its inference rules satisfy the multiple consequence property, if there is a SAT circular proof of clause $A$ from $\F$ under SAT proof system $\textbf{S}$ then $\F\models A$.
\end{theorem}

\begin{theorem} (see Theorem 4 in~\cite{DBLP:conf/sat/AtseriasL19})
There is a circular refutation of polynomial length of $\mathit{PHP}$ using the proof system with symmetric resolution and split.
\end{theorem}

 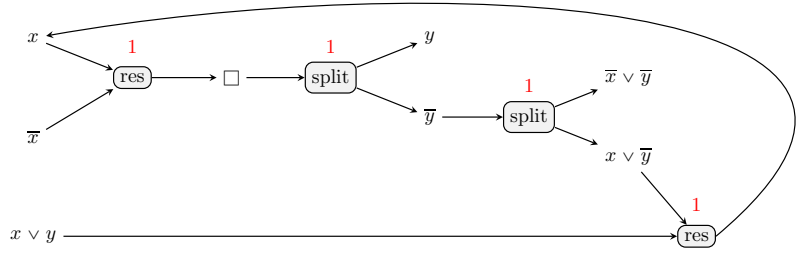
\begin{figure}
 \centering
\scalebox{0.75}{\begin{tikzpicture}
\tikzset{%
every node/.style={fill=white, rounded corners, node distance=5em},%
every path/.style={->, >=stealth, line width=.06em}%
}

\tikzstyle{op}=[draw, fill=black!5!white]

\node(dummy){};  % para centrarlo más porque el arco circular lo estropea
\node[right of=dummy, xshift=1.5cm](x){$x$};
\node[below of = x](-x){$\myneg x$};
\node[below of = -x](xy){$x \lor y$};

\node[op, right of= x, yshift=-0.7cm](res1){res};
\node[above of=res1, yshift=-1.2cm]{{\color{red}1}};
\draw (x) to (res1);
\draw (-x) to (res1);

\node[right of= res1](box){$\Box$};
\draw (res1) to (box);

\node[op, right of= box](split1){split};
\node[above of=split1, yshift=-1.2cm]{{\color{red}1}};
\draw (box) to (split1);
\node[right of= split1, yshift=0.7cm](y){$y$};
\node[right of= split1, yshift=-0.7cm](-y){$\myneg y$};
\draw (split1) to (y);
\draw (split1) to (-y);

\node[op, right of= -y](split2){split};
\node[above of=split2, yshift=-1.2cm]{{\color{red}1}};
\draw (-y) to (split2);
\node[right of= split2, yshift=0.7cm](-y-x){$\myneg x \lor \myneg y$};
\node[right of= split2, yshift=-0.7cm](-yx){$x\lor \myneg y $};
\draw (split2) to (-y-x);
\draw (split2) to (-yx);

\node[op, right of= xy, xshift=10cm](res2){res};
\node[above of=res2, yshift=-1.2cm]{{\color{red}1}};
\draw (xy) to (res2);
\draw (-yx) to (res2);

\draw (res2.east) to [out=40, in=10, looseness=1.4] (x);

\end{tikzpicture}}
\caption{Graph  of a circular proof of $\{y\}$ from $\{(x \lor y), (\myneg x)\}$. The certifying flow is indicated above each inference node.}
\label{fig:circular}
\end{figure}

Figure~\ref{fig:circular} shows the graph and certifying flow of a circular proof of $\{y\}$ from $\{x \lor y, \myneg x\}$ with symmetric resolution and split.

\subsection{Relation between ResSV and circular proofs}

Now we show that the MaxSAT \textbf{ResSV} proof system is an extension of circular proofs from SAT to MaxSAT. The following two theorems show that, when restricted to hard formulas, \textbf{ResSV} and SAT circular can simulate each other. Recall from Corollary~\ref{coro-1} that if $\F$ is a hard formula, then $\F \models \{(A,\infty)\}$ is equivalent to $MaxSAT(\F \cup \{(\myneg A,1)\})\geq 1$ which can be proved by a derivation $\F \cup  \{(\myneg A,1)\} \vdash \{(\Box,1)\}$. 

\begin{theorem}\label{th-circular2ressv}
Let $\Pi$ be a SAT circular proof of clause $A$ from formula $\F=\{C_1,\ldots,C_p\}$ using the proof system  symmetric resolution and split. There is a  proof $\{(C_1,\infty),\ldots, (C_p,\infty), (\myneg A,1)\} \vdash_{ResSV} \{(\Box,1)\}$ whose length is $O(|\Pi|)$.
\end{theorem}

\begin{proof} Let $G(\Pi)=(J\cup I,E)$  be the proof graph and $f(\cdot)$ be the flow of $\Pi$. By definition of SAT circular proof, $A\in J$ and $b(A)>0$. Let $o_C =\sum_{R\in N^+(C)} f(R)$ denote the outflow of every clause $C\in J$ and $i_C =\sum_{R\in N^-(C)} f(R)$ denote the inflow of every clause $C \in J$.

First, we show that there is a proof (made exclusively of virtual and unmerge steps),
$$\emptyset; \ldots; \{(C,-o_C) \mid C\in J\} \cup \{(C,f(R)) \mid C\in J, R\in N^+(C)\}$$
This is obtained by considering each clause node $C\in J$, adding $\{(C,-o_C), (C,o_C)\}$ thanks to the virtual rule, and unmerging $(C,o_C)$ as needed. Note that, after these steps, all the antecedents of the inference nodes in $I$ are available.

Second, we show that there is a proof (made exclusively of splits and symmetric resolutions),
$$\{(C,f(R)) \mid C\in J, R\in N^+(C)\} \vdash_{ResS} \{(C,f(R)) \mid C\in J, R\in N^-(C)\}$$
This is obtained by considering each inference node $R\in I$ and transforming its antecedents into its consequences as follows. If $R$ is a SAT split $\{C\} \vdash \{C \lor x, C \lor \myneg x\}$ then the inference step is a MaxSAT split $\{(C, f(R))\} \vdash \{(C \lor x, f(R)), (C \lor \myneg x, f(R))\}$. If $R$ is a symmetric SAT resolution $\{C \lor x, C' \lor \myneg x\} \vdash \{C \lor C'\} $ then the inference step is a (symmetric) MaxSAT resolution $\{(C \lor x, f(R)), (C' \lor \myneg x, f(R))\} \vdash \{ (C \lor C', f(R))\}$.

From the previous two proofs, %one can see the existence of a proof,
$$\F \cup (\myneg A,1); \ldots; \F \cup \{(\myneg A,1)\} \cup \{(C,-o_C) \mid C\in J\} \cup \{(C, i_C) \mid C \in J\};$$ $$ \ldots; \F \cup \{(\myneg A,1)\} \cup \{(C, b(C)) \mid C\in J\}\hspace{1.3cm} $$

%$$\F \cup (\myneg A,1) \vdash_{ResSV} \F \cup (\myneg A,1) \cup (A,f(R))$$

%since all the clauses with negative weights are cancelled out...

Let $A = a_1 \lor a_2 \lor \ldots \lor a_q$. Since, $A \in J$, $b(A) > 0$ and $(\myneg A,1)$ is shorthand for $\{(\myneg a_1,1), (a_1 \lor \myneg a_2,1), \ldots, (a_1 \lor \ldots \lor a_{q-1} \lor \myneg a_q, 1)\}$, after $q$ MaxSAT resolutions $\{(\myneg A,1), (A, 1), (A, b(A) - 1)\} \vdash_{ResSV} \{(\Box, 1)\}$ which proves the Theorem.

\end{proof}

\begin{lemma}\label{lema-infinity}
% ===================================================================
% JAVIER, he cambiado un poco la demostración utilizando el grafo de la refutación.  La idea es la misma. Lo tengo que pulir un poquito. Lo hago a lo largo del fin de semana.
% ===================================================================
Given a formula $\F=\{(C_1, \infty), \ldots, (C_j, \infty), (C_{j+1}, w_{j+1}), \ldots, (C_p, w_p)\}$ where $\forall j < k \leq p$, $w_k \not= \infty$, if there exists a  \textbf{ResSV} refutation $\F \vdash_{ResSV} \{(\Box, 1)\}$ of length $l$ then there exits a \textbf{ResSV} refutation $\F' \vdash_{ResSV} \{(\Box, 1)\}$ of  length $O(l)$ where $$\F'=\{(C_1, w_1), \ldots, (C_j, w_j), (C_{j+1}, w_{j+1}), \ldots, (C_p, w_p)\}$$
\noindent and $\forall 1 \leq k \leq j$, $w_k \not= \infty$.
\end{lemma}
\begin{proof}
For readability reasons, $\vdash$ denotes $\vdash_{ResSV}$. By Property~\ref{prop-symres}, any $\F \vdash \{(\Box, 1)\}$ of length $l$ can be rewritten into an equivalent refutation of length $e=O(l)$ in which resolution is restricted to its symmetric form. Let $\Pi = (\F_0; \F_1; \ldots; \F_e)$ be that refutation where $\F_0 = \F$ and $(\Box, 1)\in \F_e$.
%, and let $\G(\Pi) = (J \cup I, E)$ be its refutation graph. 

%We will traverse the refutation graph in reverse order in which inference nodes where added (i.e., from $\F_e$ to $\F_0$) transforming clauses nodes in $J$ with infinite weights into finite ones. At each step $i$, all inference steps $\F_i; \ldots; \F_e$ are still valid and applied over finite weighted clauses. Recall that clause nodes only have one out-neighbour and, infinitive weighted clauses is a consequent each time it is used as an antecedent.

We are going to prove that for each $\F_i= \{(C_1, \infty), \ldots, (C_j, \infty), (C_{j+1}, w_{j+1}),\ldots, (C_p, w_p)\}$ there is a $\F'_i= \{(C_1, w_1), \ldots, (C_j, w_j), (C_{j+1}, w_{j+1}),\ldots, (C_p, w_p)\}$ such that $\forall 1\leq k \leq j$ $w_k\not=\infty$ and
 $\F_i' \vdash \{(\Box, 1)\}$.  We prove it by induction on $i$ going in reverse order, from $i=e$ to $i=0$.

%$\F_i' \vdash \{(\Box, 1)\}$ by doing the same last $i$ inference steps as in $\F_i; \ldots; \F_e$. Moreover, those inference rules will be applied over finite weighted clauses.  

%transforming clauses nodes in $J$ with infinite weights into finite ones. At each step $i$, all inference steps $\F_i; \ldots; \F_e$ are still valid and applied over finite weighted clauses. 

%Recall that clause nodes only have one out-neighbour and, infinitive weighted clauses is a consequent each time it is used as an antecedent.

%infinite weights in the graph into finite ones (i.e., $\forall_{(C, u) \in \F_i} u \not= \infty$) ensuring that, at each step $i$, $\F_i; \ldots; \F_e$ is still a valid refutation and the inference rule is applied over clauses with finite weight. 
%Note that for all $0 \leq i \leq e$, $\{(C_1, \infty), \ldots, (C_j, \infty)\} \subseteq \F_i$ because all rules replace antecedents with consequents and, in all rules, infinity weighted antecedent clauses are always consequent clauses. However, transforming the weight of one of those clauses in $\F_i$ does not affect its weight in $\F_k$ $k > i$.

%$MaxSAT(\F_i) = 1$. 

\underline{Base case} ($i=e$): we define,
$$\F_e' = \{(C_1, 1), \ldots, (C_j, 1), (C_{j+1}, w_{j+1}), \ldots, (C_p, w_p)\}$$ 
Since $(\Box, 1) \in \F_e$, then $(\Box, 1) \in \F'_e$. Thus, it
satisfies the conditions.
%\noindent and $\G_e' = \F_e'$.

\underline{Inductive step}:
Let $\A \subseteq \F_{i - 1}$ and $\C \subseteq \F_{i}$ be the set of antecedent and consequent clauses in step $\F_{i - 1};\F_i$. Therefore,  $\A; \C$ and $\F_{i-1} = (\F_i \setminus \C) \cup \A$. By induction hypothesis, there is a $\C' = \{(C, w) \mid (C,u) \in \C\} \subseteq \F_i'$ such that if $u \not= \infty$ then $w = u$, else $w$ is finite. We  define $\F'_{i-1}$ as $(\F_i' \setminus \C') \cup \A'$ for some $\A'$ satisfying:
\begin{enumerate}
    \item $\A' = \{(C, w) \mid (C, u) \in \A\}$ such that if $u \not= \infty$ then $w = u$, else $w$ is finite, and
    \item there is a proof $\A'; \ldots; \C'$
\end{enumerate}
As a result, $\F_{i-1}'$ has the same clauses as $\F_{i-1}$ but with finite weight. Moreover, since there is a proof $\A'; \ldots; \C'$ where $\A' \in \F_{i-1}'$ and $\C'\in \F_i'$, then there is a proof $\F_{i-1}' ; \ldots; \F_i'$ and, since by induction $\F_i' \vdash \{(\Box, 1)\}$, then $\F_{i-1}' \vdash \{(\Box, 1)\}$.

Next, we show how to obtain such $\A'$ for the different cases. If all clauses in $\A$ have finite weight then all clauses in $\C$ have also finite weight. As a consequence, $\C' = \C$. Then, $\A'=\A$ trivially satisfies the conditions. If some clause in $\A$ has infinite weight then we analyze each possible inference rule that can happen in the $\F_{i - 1};\F_i$ step:

\begin{itemize}
    %\item Split: if the antecedent clause has finite weight then $B = B'$ and $\F_{i - 1}' = \F_i' \setminus B \cup A$. Otherwise, 
    \item \emph{Split}: By definition of the split rule, $\A = \{(C, \infty)\}$, $\C = \{(C \lor x, \infty), (C \lor \myneg x, \infty)\}$). Besides, $\C'=\{(C \lor x, u), (C \lor \myneg x, v)\}$. Then, $\A'=\{(C, \max\{u, v\})\}$ satisfies the conditions.
    %. As a consequence, $\{(C, \max\{u, v\})\} \vdash \{(C \lor x, \max\{u, v\}), (C \lor \myneg x, \max\{u, v\})\} \vdash B'$ and, therefore, $\F_{i-1}' \vdash \F_i'$.
    
    %\item Symmetric MaxSAT: if both antecedents have finite weights then $B = B'$ and $\F_{i - 1} = \F_i \setminus B \cup A$. Otherwise, 
    
    \item \emph{Symmetric resolution}: By definition of the symmetric resolution rule, $\A = \{(C \lor x, \infty), (C \lor \myneg x, \infty)\}$, $\C= \{(C, \infty), (C \lor x, \infty), (C \lor \myneg x, \infty)\}$. Besides, $\C' = \{(C, u), (C \lor x, v), (C \lor \myneg x, v')\}$. Then,  $\A'=\{(C \lor x, u + v), (C \lor \myneg x, u + v')\}$ satisfies the conditions. 
    
    %As a consequence, $\{(C \lor x, u + v)\} \vdash \{(C \lor x, u), (C \lor x, v)\}$, $\{(C \lor \myneg x, u + v')\} \vdash \{(C \lor \myneg x, u), (C \lor \myneg x, v')\}$ and $\{(C \lor x, u), (C \lor \myneg x, u)\} \vdash \{(C, u)\}$. Therefore, $\F_{i-1}' \vdash \F_i'$. 
    
    %$A = \{(C \lor x, \cdot), (C \lor \myneg x, \cdot)\}$, $B= \{(C, \cdot), (C \lor x, \cdot), (C \lor \myneg x, \cdot)\}$ and $B' = \{(C, u), (C \lor x, v), (C \lor \myneg x, v')\}$. Then,  $\F_{i-1}' = \F_i' \setminus B' \cup \{(C \lor x, u + v), (C \lor \myneg x, u + v')\}$. Note that if only one of the antecedent clauses has finite weight (e.g., $(C \lor x, w)$), then $w = u + v$. As a consequence, $\{(C \lor x, u + v)\} \vdash \{(C \lor x, u), (C \lor x, v)\}$, $\{(C \lor \myneg x, u + v')\} \vdash \{(C \lor \myneg x, u), (C \lor \myneg x, v')\}$ and $\{(C \lor x, u), (C \lor \myneg x, u)\} \vdash \{(C, u)\}$. Therefore, $\F_{i-1}' \vdash \F_i'$.

    %\item Merge: if both antecedents have finite weight then $B = B'$ and $\F_{i - 1} = \F_i \setminus B \cup A$. 
    
    \item \emph{Merge} with both antecedents having infinite weight: By definition of merge rule, $\A = \{(C, \infty),$ $(C, \infty)\}$, $\C = \{(C, \infty)\}$. Besides, $\C' = \{(C, v)\}$. Then, $\A'=\{(C, v), (C, v)\}$ satisfies the conditions. 
    
    \item \emph{Merge} with one of the antecedents having finite weight: By definition of merge rule, $\A = \{(C, \infty), (C, u)\}$, $\C = \{(C, \infty)\}$. Besides, $\C' = \{(C, v)\}$:
    \begin{itemize}
    \item if $0 < v \leq u$, then $\A' = \{(C, 1), (C, u)\}$ satisfies the conditions.
    
    %$\F_{i-1}' = \F_i' \setminus B' \cup \{(C, 1), (C, u)\}$. As a consequence, $\{(C, u)\} \vdash \{(C, v), (C, u - v)\}$ and $\F_{i-1}' \vdash \F_i'$.
    \item otherwise, $\A' = \{(C, v - u), (C, u)\}$. Note that $u$ could be a negative weight coming from a virtual rule. In any case, $v - u > 0$ and $\A'$ satisfies the conditions. 
    
    %$\F_{i-1}' = \F_i' \setminus B' \cup \{(C, v - u), (C, u)\}$. Note that, $u$ could be a negative weight coming from a virtual rule. In any case, $v - u > 0$, $\{(C, v - u), (C, u)\} \vdash \{(C, v)\}$ and $\F_{i-1}' \vdash \F_i'$.
    
    %\item if $u < 0$, then $\F_{i-1}' = \F_i' \setminus B' \cup \{(C, v - u), (C, u)\}$ 
    %\item else if $v > u$, then $\F_{i-1}' = \F_i' \setminus B' \cup \{(C, v - u), (C, u)\}$
    %\item otherwise, $\F_{i-1}' = \F_i' \setminus B' \cup \{(C, 1), (C, u)\}$. As a consequence, $\{(C, u)\} \vdash \{(C, v), (C, u - v)\}$ and $\F_{i-1}' \vdash \F_i'$.
    %\item otherwise, 
    
    %(XXX aqui al reproducir la proof tendra que hacer un split en vez de un merge XXX)
    %Then, $\F_{i-1}' = \F_i' \setminus B' \cup \{(C, \max\{v - u, 1\}), (C, u)\}$. Note that $u$ could be a negative weight coming from a virtual rule.
    \end{itemize}
    
    %the rule is $ \vdash \{(C, v)\}$. Replace $\{(C, \infty), (C, \infty)\}$ by $\{(C, v), (C, v)\}$ in $\F_{i-1}$. Otherwise, the rule is: $\{(C, \infty), (C, u)\} \vdash \{(C, v)\}$. Replace $\{(C, \infty), (C, u)\}$ by $\{(C, v - u), (C, u)\}$ in $\F'$. Note that $u$ could be a negative weight coming from a virtual rule.
    %\item Unmerge: if the antecedent has finite weight then $B = B'$ and $\F_{i - 1} = \F_i \setminus B \cup A$. Otherwise, 
    
    \item \emph{Unmerge}: By definition of unmerge rule, $\A =\{(C, \infty)\}$, $\C = \{(C, \infty), (C, \infty)\}$. Besides, $\C' = \{(C, u), (C, v)\}$. Then, $\A' = \{(C, u + v)\}$ satisfies the conditions.
    
    %the rule is: $\{(C, \infty)\} \vdash \{(C \lor x, u), (C \lor \myneg x, v)\}$. Replace $\{(C, \infty)\}$ by $\{(C, u + v)\}$ in $\F'$. 

    %do nothing because the weights are always finite.
    %\item \textbf{Virtual:} By definition of virtual $A = \emptyset$ and $B=\{\}$. In this case $B = B'$. Then  $\F_{i - 1}' = \F_i' \setminus B$.
\end{itemize}

%\noindent Note that, as $\F_{i - 1}'$ is constructed by replacing the consequents $B'$ in $\F_i'$ with a sufficiently large finite weighted version of its antecedents $A$, requirement $(2)$ is always fulfilled. Moreover, each step  $\F_{i-1};\F_i$ is replaced by a proof $\F_{i - 1}' \vdash \F_i'$ of at most length $3$. 

%\noindent In any of the above cases, inference step $\F_{i - 1} \vdash \F_i$ could be reproduced in $\F_{i-1}'$ because we have given its antecedents enough weight and, by induction hypothesis, all inference steps in refutation $F_i; \ldots; \F_e$ can be reproduced from $\F_i'$.
%Note that all clauses in $F_0$ have finite weight in $\F_0'$ and $\F_0' \vdash \{(\Box, 1)\}$ with length $O(l)$.

%The result is that all clauses in $\G(\Pi)$ represent finite weighted clauses. $F'$ are the ones with none in-neighbours.

%The result is that all clauses in $\F_0$ has been transformed into a set of finite weighted clauses $\F'$ such that refutation $\F' \vdash_{ResSV} \{(\Box, 1)\}$ has length $e$. 
\end{proof}

\begin{theorem}\label{th-ressv2circular}
Consider a hard formula ${\cal H}=\{(C_1,\infty),\ldots,(C_p,\infty)\}$ and a MaxSAT proof ${\cal H}  \cup \{(\myneg A,1)\} \vdash_{ResSV} \{(\Box,1)\}$ of length $e$. There is a SAT circular proof $\Pi$ of $A$ from ${\cal H}'=\{C_1,\ldots,C_p\}$ with proof system having symmetric resolution and split. The length of the circular proof is $O(e)$.
\end{theorem}
\begin{proof}
From derivation ${\cal H}  \cup \{(\myneg A,1)\} \vdash_{ResSV} \{(\Box,1)\}$ we need to build a pre-proof $\Pi$ with a (possibly cyclic) graph $G(\Pi)=(J\cup I, E)$ and a flow $f(\cdot)$ that certifies that the pre-proof is indeed a circular proof. The graph must satisfy that ${\cal H}' \subset J$, $A\in J$; its inference nodes must be consistent with either symmetric resolution or split. Also, the flow $f(\cdot)$ must satisfy the balance conditions including that $A$ has strictly positive balance.

First, by Lemma~\ref{lema-infinity}, there exists an $\F = \{(C_1,w_1),\ldots,(C_p,w_p)\}$ with $w_k \not= \infty$ forall $1 \leq k \leq p$, such that $\F \cup \{(\myneg A,1)\} \vdash_{ResSV} \{(\Box,1)\}$ with length $O(e)$ where resolution is restricted to its symmetric form. Moreover, since the virtual rule does not have antecedents all its applications can be done at the beginning of the derivation and all the cancellation of all the virtual clauses can be done at the end. 
Therefore, $\F  \cup \{(\myneg A,1)\} \vdash_{ResSV} \{(\Box,1)\}$ implies the existence of a derivation $\Gamma$, 
$$\F_0 = (\F \cup \{(\myneg A, 1)\} \cup \B); \F_1; \F_2; \ldots; (\G \cup \{(\Box, 1)\} \cup \B); ...;\F_{m-1}; (\G \cup \{(A, 1), (\myneg A, 1)\} \cup \B) = \F_m$$

\noindent where $\B$ is the set of clauses with positive weight added by the virtual rule in the original \textbf{ResSV} derivation, $m = O(e)$, and the only inference rules needed are split and symmetric resolution (along with the usual merge and unmerge).

First, we build the (acyclic) graph $G(\Gamma)$  along with a flow function $f(\cdot)$. 
Let $G_i(\Gamma) = (J_i \cup I_i, E_i)$ be the graph at step $i$, $b_i(C)$ be the balance of node $C \in J_i$ in $G_i(\Gamma)$, and let $merged(\F_i)$ be equivalent to $\F_i$ with no repeated clauses.

We will traverse the derivation $\Gamma$ from $\F_0$ to $\F_m$ ensuring that, at each step $i$, $G_i(\Gamma)$ satisfies:
\begin{enumerate}
    \item $\forall (C, w) \in \F_i$, $C \in J_i$
    \item $\forall (C, w) \in merged(\F_i)$, $w = b_i(C)$
    \item all nodes in $J_i$ are different
\end{enumerate}
We proceed by induction on the step $i$.

Base case ($i = 0$). Then:
\begin{itemize}
\item[$-$] $J_0 = \{C \mid (C, w) \in merged(\F_0)\}$
\item[$-$] $I_0 = \{d_C \mid C \in J_0\}$ (dummy inference nodes)
\item[$-$] $E_0 = \{(d_C, C) \mid C \in J_0, d_C \in I_0\}$
\item[$-$] $\forall {(C, w) \in merged(\F_0)}, f(d_C) = w$
\end{itemize}

Inductive step: Let $\A \subseteq \F_i$ and $\C \subseteq \F_{i+1}$ be the antecedents and consequents of $\F_i; \F_{i+1}$ respectively. Note that, by induction hypothesis for every clause $(C, w) \in \A$ there is a node $C \in J_i$. The construction of $G_{i+1}(\Gamma)$ depends on the inference rule used:

\begin{itemize}
    \item Split/Symmetric resolution: 
    
        \begin{itemize}
            \item $J_{i+1} = J_i \cup \{C \mid (C, w) \in \C \land C \not\in J_i\}$
            \item $I_{i+1} = I_i \cup \{i+1\}$
            \item $E_{i+1} = E_i \cup \{(C, i + 1) \mid (C, w) \in \A\} \cup \{(i+1, C) \mid (C,w)\in \C\}$
            \item $f(i+1) = w$, where $w$ is the common weight of all clauses in $\A$
        \end{itemize}
    As a result, $\forall (C, w) \in \A, b_{i+1}(C) = b_i(C) - w$ and $\forall (C, w) \in \C, b_{i+1}(C) = b_i(C) + w$. Since $\forall (C, w) \in \A$, its weight in $merged(\F_{i+1})$ is decreased by $w$ wrt its weight in $merged(\F_i)$, and $\forall (C, w) \in \C$, its weight in  $merged(\F_{i+1})$ is increased by $w$ wrt its weight in $merged(\F_i)$, we can guarantee that $G_{i+1}(\Gamma)$ satisfies (1), (2) and (3).

    %\item Split: ($(C, w) \vdash (C \lor x, w), (C \lor \myneg x, w)$). Then, $J_{i + 1} = J_i \cup \{C \lor x, C \lor \myneg x\}$ (solo la que no exista); $I_{i + 1} = I_i \cup \{i + 1\}$; $E_{i + 1} = E_i \cup \{(C, i+1), (i+1, C \lor x), (i + 1, C \lor \myneg x)\}$; $f(i+1) = w$. As a consequence, $b_{i + 1}(C) = b_i(C) - w$, $b_{i +1}(C \lor x) = b_i(C \lor x) + w$, $b_{i +1}(C \lor \myneg x) = b_i(C \lor \myneg x) + w$ and $\G_{i + 1}$ satisfy (1), (2) and (3).
    %\item Symmetric resolution: ($(C \lor x, w), (C \lor \myneg x, w) \vdash (C, w)$). Then, $J_{i + 1} = J_i \cup \{C\}$ (solo si no existe); $I_{i + 1} = I_i \cup \{i + 1\}$; $E_{i + 1} = E_i \cup \{(C \lor x, i+1), (C \lor \myneg x, i+1),(i+1, C)$; $f(i+1) = w$. As a consequence, $b_{i+1}(C \lor x) = b_i(C \lor x) - w$, $b_{i+1}(C \lor \myneg x) = b_i(C \lor \myneg x) - w$, $b_{i+1}(C) = b_i(C) + w$ and  $\G_{i + 1}$ satisfy (1), (2) and (3).
    \item Merge/Unmerge: since $merged(\F_{i+1}) = merged(\F_i)$, we define $G_{i + 1}(\Gamma)$ as  $G_i(\Gamma)$ which, by induction hypothesis, satisfies (1), (2) and (3).
\end{itemize}

The result is that there is a node in $G_m(\Gamma)$ for all clauses in $merged(\F_m)$ and the weight of each of them corresponds to its balance. In particular, ${\cal H} ' \in J_m$; $A \in J_m$ and $b_m(A) \geq 1$; $\forall (C, w) \in \B$, $b_m(C) \geq w$; and $\forall C \in \myneg A$, $C \in J_m$ and  $b_m(C) \geq 1$.

Let $G(\Pi) = (J_p \cup (I_p \setminus I_0), E_p \setminus E_0)$. Since $\forall (C, w) \in merged(\F_0)$, $d_C \in I_0$ and $f(d_C) = w$:
\begin{itemize}
\item $\forall (C, w) \in \B$,  $b(C) = b_m(C) - w \geq 0$;
\item $\forall C \in {\cal H}'$,  $b(C)$ may become negative (but they are the hard clauses);
\item $\forall C\in \myneg A$, $b(C) = b_m(C) - 1 \geq 0$. 
\end{itemize}

\noindent Moreover, balance $b(A)$ remains positive.
\end{proof}

\begin{comment}
\begin{lemma}
Given a formula $\F= \{(C_1, \infty), \ldots, (C_p, \infty), (\myneg A, 1)\}$, if there is a ResSV refutation $\F \vdash_{ResSV} \{(\Box, 1)\}$ using symmetric resolution of length $e$, then $\F' \vdash_{ResSV} \{(\Box, 1)\}$ where $\F' = \{(C_1, w_1), \ldots, (C_p, w_p), (\myneg A, 1)\}$ and $w_i \leq m*O(e)!$, $m$ being the maximum number of times an original clause $(C_i, \infty)$ is used along refutation $\F \vdash_{ResSV} \{(\Box, 1)\}$.
\end{lemma}
\begin{proof}
By Theorem~\ref{}, $\F \vdash_{ResSV} \{(\Box, 1)\}$ means that there is a SAT circular proof $\Pi$ of length $O(e)$. By Lemma 1 in~\cite{DBLP:conf/sat/AtseriasL19}, the flow of each inference vertex in every SAT circular proof $\Pi$ is bounded by a positive integer $O(e)!$. By Theorem~\ref{}, every SAT circular proof $\Pi$ can be rewritten as a ResSV proof where the flow of inference nodes are the weights moved by inference rules. Therefore, if $m$ is the maximum number of times an original hard clause in $\Pi$ is used (i.e., the size of its out-neighbours), ....
\end{proof}
\end{comment}

\section{Related Work} \label{Sec-Related}
In this Section we review and discuss some works in chronological order that have influenced the research presented in this paper.

\subsection{Soft Probing}\label{Sec-softprobing}

\textbf{ResSV} contains three rules that provide increasing refutational power. While increasing the power is a desirable feature, having more rules to choose from makes the automatization more difficult. Therefore, one practical challenge is to use split and virtual in a controlled but potentially useful way. Soft Probing is a technique that was used as a pre-process in the MiniMaxSAT solver~\cite{DBLP:journals/jair/HerasLO08} to extract an initial lower bound from MaxSAT formulas. It can be seen as a simple, yet efficient implementation of this idea. In the original paper, the technique is presented algorithmically and very briefly. Next, we show how it fits into the context of this paper. 

Consider the following theorem,
\begin{theorem}
Let $\F$ be a weighted MaxSAT formula. If there is a unary (i.e, made exclusively of unit clauses) formula ${\cal U}$ such that if $(l,w)\in {\cal U}$ and $k=\sum_{(l,w)\in {\cal U}} w$ then $(\myneg l,u)\notin {\cal U}$, and,
\begin{enumerate}
    \item $\F \cup {\cal U} \vdash_{Res} \G \cup \{(\Box, k)\} $
    \item $\G \cup \myneg {\cal U} \vdash_{Res} \{(\Box, k')\}$ 
\end{enumerate}
Then, $\F \vdash_{ResSV} \{(\Box, k')\}$.
\end{theorem}
\begin{proof}
From $\F$ we apply the virtual rule with every unit clause in ${\cal U}$ obtaining $\F \cup {\cal U} \cup {\cal U}^-$ with ${\cal U}^-=\{(l,-w) \mid (l,w)\in {\cal U}\}$. Then, we use the first proof in the theorem obtaining
$$\G \cup {\cal U}^- \cup \{(\Box, k)\}$$
Using the split rule, we transform $(\Box, k)$ into ${\cal U} \cup \myneg {\cal U}$ obtaining
$$\G \cup {\cal U}^- \cup {\cal U} \cup \myneg {\cal U}$$
Then we eliminate ${\cal U}^- \cup {\cal U}$ which cancel each other and use the second proof of the theorem to obtain
$$\{(\Box,k')\}$$
\end{proof}

This Theorem gives a \textbf{Res} condition to identify a \textbf{ResSV} derivation that produces an increment in the lower bound. 
Soft Probing applies this theorem iteratively for every literal $l$ in the formula. At each step, ${\cal U}$ is restricted to $\{(l,w)\}$ and it only considers unit propagation (which can be implemented efficiently) for the two derivations.

%restricted to ${\cal U}=\{(l,w)\}$ and only considering unit propagation (which can be implemented efficiently) for the two derivations. It iterates the process for every literal in the formula.

Now, a natural question arises: how powerful is \textbf{ResSV} when restricted to the use of this Theorem? Interestingly enough, it is sufficient for refuting the $PHP$ and $SPHP$ in polynomial time. The following property shows that both problems satisfy the conditions of the previous theorem.

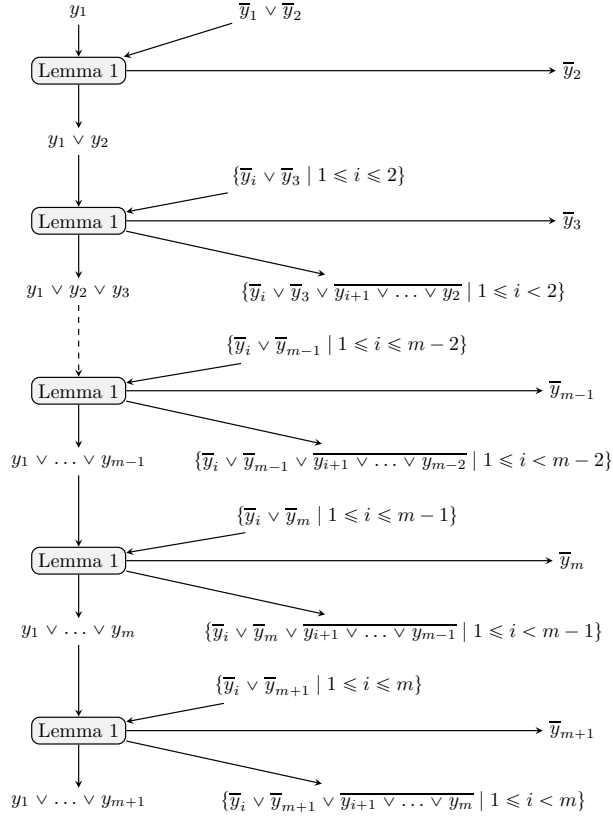
\begin{figure}[t]
    \centering
\scalebox{0.75}{\begin{tikzpicture}

\tikzset{%
every node/.style={fill=white, rounded corners, node distance=5em},%
every path/.style={->, >=stealth, line width=.06em}%
}

\tikzstyle{op}=[draw, fill=black!5!white]

\node(-x){$y_1\lor \ldots \lor y_{m+1}$};
\node[right of=-x,xshift=4cm](-nr){$\{ \myneg y_i \lor \myneg y_{m+1} \lor \myneg{y_{i+1} \lor \ldots \lor y_{m}} \mid 1 \leq i < m\} $};

\node[op, above of=-x, yshift=-0.5cm](op-x){Lemma~\ref{le1}};
\node[right of=op-x,xshift=7cm](-n){$\myneg y_{m+1} $};

\draw (op-x) to (-x);
\draw (op-x) to (-n);
\draw (op-x) to (-nr);

\node[right of=op-x, xshift=2.5cm, yshift=0.8cm](-k1-k){$\{ \myneg y_i \lor \myneg y_{m+1} \mid 1 \leq i \leq m\}$};
\node[above of=op-x](k1-k){$y_1 \lor \ldots \lor y_{m}$};
\node[right of=k1-k,xshift=4cm](k1-k-rr){$\{ \myneg y_i \lor \myneg y_{m} \lor \myneg{y_{i+1} \lor \ldots \lor y_{m-1}} \mid 1 \leq i < m- 1\} $};
\draw (-k1-k) to (op-x);
\draw (k1-k) to (op-x);

\node[op, above of=k1-k, yshift=-0.5cm](opk1-k-r){Lemma~\ref{le1}};
\node[right of=opk1-k-r,xshift=7cm](k1-k-r){$\myneg y_{m} $};

\draw (opk1-k-r) to (k1-k);
\draw (opk1-k-r) to (k1-k-r);
\draw (opk1-k-r) to (k1-k-rr);

\node[right of=opk1-k-r, xshift=3cm, yshift=0.8cm](-k2-k){$\{ \myneg y_i \lor \myneg y_{m} \mid 1 \leq i \leq m - 1\}$};
\node[above of=opk1-k-r](k2-k){$y_1 \lor \ldots \lor y_{m-1}$};
\node[right of=k2-k,xshift=4cm](k2-k-rr){$\{ \myneg y_i \lor \myneg y_{m-1} \lor \myneg{y_{i+1} \lor \ldots \lor y_{m-2}} \mid 1 \leq i < m- 2\} $};
\draw (k2-k) to (opk1-k-r);
\draw (-k2-k) to (opk1-k-r);

\node[op, above of=k2-k, yshift=-0.5cm](opk2-k){Lemma~\ref{le1}};
\node[right of=opk2-k, xshift=7cm](k2-k-r){$\myneg y_{m-1}$};

\draw (opk2-k) to (k2-k);
\draw (opk2-k) to (k2-k-r);
\draw (opk2-k) to (k2-k-rr);

\node[right of= opk2-k, xshift=3cm, yshift=0.8cm] (-k3-k) {$\{ \myneg y_i \lor \myneg y_{m-1} \mid 1 \leq i \leq m-2\}$};
\node[above of=opk2-k](3-k){$y_1 \lor y_2 \lor y_3$};
\node[right of=3-k,xshift=4cm](3-k-rr){$\{ \myneg y_i \lor \myneg y_{3} \lor \myneg{y_{i+1} \lor \ldots \lor y_{2}} \mid 1 \leq i < 2\} $};
\draw[dashed] (3-k) to (opk2-k);
\draw (-k3-k) to (opk2-k);

\node[op, above of=3-k, yshift=-0.5cm](op3-k){Lemma~\ref{le1}};
\node[right of=op3-k, xshift=7cm](3-k-r){$\myneg y_3$};

\draw (op3-k) to (3-k);
\draw (op3-k) to (3-k-r);
\draw (op3-k) to (3-k-rr);

\node[right of=op3-k, xshift=2.5cm, yshift=0.8cm](-2-k){$\{ \myneg y_i \lor \myneg y_3 \mid 1 \leq i \leq 2\}$};
\node[above of=op3-k, node distance=4em](2-k){$y_1 \lor y_2$};
\draw (-2-k) to (op3-k);
\draw (2-k) to (op3-k);

\node[op, above of=2-k, yshift=-0.5cm](op2-k){Lemma~\ref{le1}};
\node[right of=op2-k, xshift=7cm](2-k-r){$\myneg y_2$};

\draw (op2-k) to (2-k);
\draw (op2-k) to (2-k-r);

\node[above of=op2-k, node distance=3em](1k1){$y_{1}$};
\node[right of=1k1, node distance=4em, xshift=2cm](-1-k){$\myneg y_1 \lor \myneg y_2$};
\draw (-1-k) to (op2-k);
\draw (1k1) to (op2-k);

\end{tikzpicture}}
\caption{Derivation graph corresponding to hole $j$. For clarity purposes, we rename each variable $x_{ij}$, $1 \leq i \leq m + 1$ to $y_i$. All clauses have cost 1.}
    \label{fig:sphp1hole}
\end{figure}

 \begin{property}
 Consider the $PHP$ and $SPHP$ problems and let ${\cal U}=\{(x_{11},1), (x_{12},1), \ldots, (x_{1m},1)\}$.
 \begin{itemize}
     \item There is a proof $PHP \cup {\cal U} \vdash_{Res} PHP \cup \{(\Box,m)\}$
     \item There is a proof $PHP  \cup \myneg {\cal U} \vdash_{Res} \{(\Box,1)\}$
     \item There is a proof $SPHP \cup {\cal U} \vdash_{Res} \G \cup \{(\Box,m)\}$
     \item There is a proof $\G  \cup \myneg {\cal U} \vdash_{Res} \{(\Box,1)\}$
 \end{itemize}
 \end{property}
\begin{proof}
First, we prove the $SPHP$ case. The first refutation of $SPHP$ is as follows. First, for each hole $j$ and $\{(x_{1j},1)\}$ there is a derivation of $\{(\myneg x_{ij},1) \mid 2 \leq i \leq m+1\}$ (see Figure~\ref{fig:sphp1hole}). Then, for each pigeon $i>1$ and $\{(\myneg x_{ij}, 1) \mid 1 \leq j \leq m\}$, there is a derivation of $\{(\Box, 1)\}$ (see Figure~\ref{fig:sphp1} (left)). Therefore, concatenating the previous derivations we get,

$$SPHP \cup {\cal U} \vdash_{Res} {\cal G} \cup \{(\Box, m)\} $$

\noindent where clause $\{(x_{11} \lor x_{12} \lor \ldots \lor x_{1m}, 1)\} \in {\cal G}$. Figure~\ref{fig:sphp1} (left) shows the derivation graph of the second refutation,

$$\{(x_{11} \lor x_{12} \lor \ldots \lor x_{1m}, 1)\} \cup \myneg {\cal U} \vdash_{Res} \{(\Box, 1)\}$$

\noindent which completes the proof.

Let us now prove the $PHP$ case. Since $PHP \vdash_{Res} PHP \cup SPHP$ by unmerging each hard clause $(C, \infty) \in PHP$ into $(C, \infty), (C, 1)$ and we have proved that $SPHP \cup  {\cal U} \vdash_{Res} \G \cup \{(\Box,m)\}$, then $PHP \cup {\cal U} \vdash_{Res} PHP \cup \{(\Box,m)\}$. Since we have proved that $\{(x_{11} \lor x_{12} \lor \ldots \lor x_{1m}, 1)\} \cup \myneg {\cal U} \vdash_{Res} \{(\Box, 1)\}$, unmerging $(x_{11} \lor x_{12} \lor \ldots \lor x_{1m}, \infty) \in PHP$ into $(x_{11} \lor x_{12} \lor \ldots \lor x_{1m}, \infty), (x_{11} \lor x_{12} \lor \ldots \lor x_{1m}, 1)$ completes the proof. 

%and $PHP  \cup \myneg {\cal U} \vdash_{Res} \{(\Box,1)\}$.

% and $\G  \cup \myneg {\cal U} \vdash_{Res} \{(\Box,1)\}$
\end{proof}

\subsection{OSAC}
\textit{Weighted Constraint Satisfaction Problems} (WCSPs) are optimization problems defined by a network of local cost functions defined over discrete variables. Thus, MaxSAT can be seen as a particular type of WCSP where the local cost functions are the clauses and variables are boolean\cite{DBLP:conf/cp/GivryLMS03}. WCSP solvers compute lower bounds by enforcing \textit{local consistency}. This is achieved by moving costs around the network using two equivalence preserving operations: \textit{projection} and \textit{extension}. WCSP projection is similar to MaxSAT symmetric resolution and WCSP extension is similar to 
split. The main difference is that in the WCSPs movements are restricted to pre-defined subsets of variables (i.e, the scopes of the original cost functions), while in \textbf{ResSV} the proof system gives complete freedom on the variables involved in the clauses. This freedom is needed to guarantee completeness, which is not a problem in the WCSP context where local consistency is not used as a stand-alone algorithm, but only as a heuristic.

Optimal Soft Arc Consistency OSAC \cite{DBLP:journals/ai/CooperGSSZW10} introduced the idea of allowing weights to become negative during the process. As in our case, it is shown that the lower bound is valid (i.e, sound) as long as all the weights are positive at the end of the process. Interestingly,  OSAC can be enforced with a linear program. Solving the linear program produces the optimal lower bound is obtained (optimal with respect to the pre-defined scopes on which costs can be moved to). 

Thus, OSAC is reminiscent to a \textbf{ResSV} proof restricting new clauses to pre-defined (and of bounded size) sets of variables. Interestingly, the efficiency of \textbf{ResSV} on the SPHP problem does not rely on the size of the clauses which is as high as the number of pigeons and holes, and therefore unbounded.

\subsection{Dual Rail Encoding}\label{Sec-dualrail}
In their recent work \cite{DBLP:conf/sat/IgnatievMM17,DBLP:conf/aaai/BonetBIMM18} introduce the dual rail encoding which transforms a SAT formula ${\cal F}$ over variables $X=\{x_1,\ldots,x_s\}$ (i.e., all clauses are hard) into a MaxSAT formula ${\cal M}$ over variables $N=\{n_1,\ldots,n_s\}$ and $P=\{p_1,\ldots,p_s\}$. The dual encoding of clause $C\in {\cal F}$ is a hard clause in which each unnegated literal $x_i$ in $C$ is replaced by $\myneg n_i$, and each negated literal $\myneg x_i$ in $C$ is replaced by $\myneg p_i$. Additionally, for each variable $x_i$ the dual encoding adds three new clauses: $(p_i,1)$, $(n_i,1)$ and $(\myneg p_i \lor \myneg n_i,\infty)$. The resulting MaxSAT formula ${\cal M}$ is made exclusively of horn clauses, where only unit clauses are soft. 

It is shown that ${\cal F}$ is satisfiable iff $s=MaxSAT({\cal M})$. They also show that $s\leq MaxSAT({\cal M})$. Accordingly, a \textit{dual rail MaxSAT refutation}, which is a proof of ${\cal F}$ unsatisfiability, is defined as a proof of $MaxSAT({\cal M})\geq \{(\Box,s+1)\}$.

 They show that there is a polynomial size proof $MaxSAT({\cal M})\vdash_{Res} \{(\Box,s+1)\}$ which indicates that the dual rail encoding makes the $PHP$ tractable \footnote{the refutation is very similar to the proof of Property~\ref{prop-Res-sphp1}} and therefore dominates the SAT resolution proof system.
 In their work it is not clear which of the dual rail ingredients (e.g. horn clauses, unit cost soft clauses, renaming,...) if not all, are really needed for this domination.
 The following Theorem shows that \textbf{ResSV} is at least as powerful as the dual encoding, which indicates that the true power of the dual encoding comes only from the introduction of the unary costs.
 
 \begin{theorem} 
 \label{th3}
 \textbf{ResSV} with variable aliases can simulate the dual rail encoding.
\end{theorem}

\begin{proof}
In the proof we allow \textbf{ResSV} to add for every original variable $x_i$ a new variable $y_i$ such that $x_i \leftrightarrow \myneg y_i$. Note that these fresh variables do not abbreviate formulas but only add variable aliases and, as a consequence, there is no gain in a proof system from adding them. 
%It is easy to see that the fresh variables are only syntactical sugar in the proof (there is no gain in a proof system from adding variable aliases) making it more intuitive. 
In the following, we show that any SAT formula can be transformed to its dual rail encoding using \textbf{ResSV} inference only. 

Let  ${\cal F}$ be a SAT formula over $X=\{x_1,\ldots x_n\}$. For each variable $x_i$, we add hard clauses $x_i\lor y_i$ and $\myneg x_i \lor \myneg y_i$, where $y_i$ is a fresh variable. The clauses only indicate that $x_i$ and $\myneg y_i$ are equivalent (i.e, no new information is added). Now, resolve each clause $x_i\lor A \in {\cal F}$ with $\myneg x_i \lor \myneg y_i$ which means that a new clause $\myneg y_i \lor A$ is added to the formula. Clearly, at the end of this process we have for each original clause $C$, a new clause $C'$ where positive literals in $C$ have been replaced by their $\myneg y_i$ equivalent.  

Next, we apply $n$ virtual rules adding at each step two fresh clauses,
\begin{center}
 \[ \begin{array}{c} \hline
(\Box, 1) \hspace{0.5cm} (\Box, -1)
\end{array}\]    
\end{center}

\noindent and then split,

\begin{center}
 \[ \begin{array}{c} 
 (\Box,1) \\\hline
 (x_i, 1) \hspace{0.5cm} (\myneg x_i, 1)
\end{array}\]    
\end{center}

\noindent for each variable $x_i$. Next, we unmerge each $(x_i\lor y_i,\infty)$ into $(x_i\lor y_i,\infty), (x_i\lor y_i, 1)$, and then we resolve each $(\myneg x_i, 1)$ with $(x_i\lor y_i, 1)$,

%we resolve each $(\myneg x_i, 1)$ with $(x_i\lor y_i,\infty)$ by first unmerging $(x_i\lor y_i,\infty)$ into $(x_i\lor y_i,\infty), (x_i\lor y_i, 1)$ and then,
\begin{center}
 \[ \begin{array}{c}
(\myneg x_i,1)\ \ \ (x_i\lor  y_i, 1)  \\
\hline
(y_i,1)\\ 
%(\myneg x_i \lor \myneg y_i, 1) \ \ \ (x_i\lor y_i, \infty)\\
\end{array}\]    
\end{center}

%\noindent where the last clause can be removed because it is subsumed by the already existing clause $(\myneg x_i \lor \myneg y_i, \infty)$.

\noindent The resulting formula contains all the clauses of the dual rail encoding, so we can simulate any dual rail refutation which, by definition, ends up generating $(\Box, n+1)$. The aggregation of the $n$ clauses $(\Box, -1)$ into $(\Box,-n)$, and then merging with $(\Box,n+1)$ produces $(\Box, 1)$. Using Corollary~\ref{coro-1} we know that this refutation proves unsatisfiability.
\end{proof}

\subsection{Equivalence between systems}\label{Sec-equiv}
In \cite{DBLP:conf/sat/BonetL20} Bonet and Levy study the equivalence between a proof system similar to \textbf{ResSV} and circular proofs. As in the dual rail encoding approach, they restrict their attention to the refutation of SAT formula. More precisely, they consider a SAT formula ${\cal F}$ and study MaxSAT refutations of the form ${\cal F}' \vdash \{(\Box, 1)\}$ where ${\cal F}'=\{(C,w) \mid C \in {\cal F}\}$, for some sufficiently large positive weight $w$.

In this setting, they independently showed that \textbf{ResSV} and SAT Circular Proofs are polynomially equivalent. This is the same as our Theorems~\ref{th-circular2ressv} and \ref{th-ressv2circular} because their choice of replacing infinities by sufficiently high weights does not affect the effectiveness of \textbf{ResSV}, as we shown in Lemma~\ref{lema-infinity}.

\section{Conclusions and Future Work}
\label{Sec-Conclusions}
Several approaches for MaxSAT solving have been proposed in the last years and most of the comparisons have been done empirically.
In this paper we set some basic definitions for a proof complexity approach, which we believe may be a very useful complement. From a descriptive point of view,
our theoretical approach provides a framework to explain under a common language some related work such as circular proofs (Section~\ref{Sec-Circular}), soft probing (Section~\ref{Sec-softprobing}) or dual rail (Section~\ref{Sec-dualrail}). Because proof systems break inferences into different rules, a proof complexity approach facilitates the understanding of the advantages and limitations of each different rule (the very recent work of \cite{DBLP:journals/eccc/FilmusMSV20} already gives some support to this claim). Our paper covers a first analysis of three inference rules: resolution, split and virtual, with split and virtual being original from our work. We show that the addition of each rule makes the proof system stronger. %provable stronger. 

We expect this work to motivate other MaxSAT practitioners to use our framework to analyze their contributions. In particular we want to explore the relationship between certifying lower bounds with search algorithms and proof systems. This idea, which has shed so much light to the SAT case would be very beneficial also for MaxSAT.  

\bibliographystyle{plain}
\bibliography{mibiblio}

\end{document}